\documentclass[11pt,letter]{article}
\usepackage[margin=1.0in]{geometry}
\usepackage[colorlinks]{hyperref}
\hypersetup{linkcolor=cyan,filecolor=cyan,citecolor=cyan,urlcolor=cyan}
\usepackage{url}
\usepackage{xspace}
\usepackage{thm-restate,color,xcolor}
\usepackage{amsmath,amssymb,bbm,amsthm}
\usepackage{enumitem}
\usepackage{fullpage}
\usepackage{boxedminipage}
\usepackage[boxed]{algorithm}
\usepackage{epigraph}
\usepackage[sc]{mathpazo}
\usepackage{amsmath}
\usepackage{mathtools}
\usepackage{framed}
\usepackage[framemethod=tikz]{mdframed}
\usepackage{titlesec}
\usepackage{lipsum}%
\usepackage{cleveref,aliascnt}
\usepackage{wrapfig}
\usepackage{enumitem}
\newcommand{\local}{${\mathsf{LOCAL}}$}
\newcommand{\congest}{${\mathsf{CONGEST}}$}

\newcommand{\dist}{\mbox{\rm dist}}

\DeclareMathOperator*{\E}{\mathrm{E}}
  \newcommand{\BBtalgo}{\mathsf{BroadcastKnownCovFamily}}                 
\newcommand{\BBalgo}{\mathsf{BroadcastKnownDiam}} 
\newcommand{\UnBBalgo}{\mathsf{Broadcast}}                                                                                  
\newcommand{\mincut}{\text{MinCut}}
\newcommand{\poly}{\mathop\mathrm{poly}}

\newcommand{\cP}{\mathcal{P}}

\newcommand{\cG}{{\mathcal G}}

\makeatletter
\newtheorem*{rep@theorem}{\rep@title}
\newcommand{\newreptheorem}[2]{%
\newenvironment{rep#1}[1]{%
\def\rep@title{#2 \ref{##1}}%
\begin{rep@theorem}}%
{\end{rep@theorem}}}
\makeatother
\newtheorem{theorem}{Theorem}

\newtheorem{corollary}[theorem]{Corollary}
\theoremstyle{plain}
\newtheorem{observation}[theorem]{Observation}
\newtheorem{lemma}[theorem]{Lemma}
\newtheorem*{lemma*}{Lemma}
\newreptheorem{lemma}{Lemma}
\theoremstyle{plain}
\newtheorem{fact}[theorem]{Fact}

\newtheorem{claim}[theorem]{Claim}
\newtheorem{question}[theorem]{Question}
\newtheorem{definition}[theorem]{Definition}

\begin{document}

\title{Broadcast CONGEST Algorithms against Adversarial Edges}
\author{
Yael Hitron \\
        \small Weizmann Institute\\
        \small yael.hitron@weizmann.ac.il
\and				
Merav Parter \thanks{This project is funded by the European Research Council (ERC) under the European Union’s Horizon 2020 research and innovation programme (grant agreement No. 949083).}\\
        \small Weizmann Institute \\
        \small merav.parter@weizmann.ac.il
}
\date{}
\maketitle

\begin{abstract}
We consider the corner-stone broadcast task with an adaptive adversary that controls a fixed number of $t$ edges in the input communication graph. In this model, the adversary sees the \emph{entire} communication in the network and the random coins of the nodes, while maliciously manipulating the messages sent through a set of $t$ edges  (unknown to the nodes).  Since the influential work of [Pease, Shostak and Lamport, JACM'80], broadcast algorithms against plentiful adversarial models have been studied in both theory and practice for over more than four decades. Despite this extensive research, there is no round efficient broadcast algorithm for \emph{general} graphs in the \congest\ model of distributed computing. Even for a \emph{single} adversarial edge (i.e., $t=1$), the state-of-the-art round complexity is polynomial in the number of nodes.

We provide the first round-efficient broadcast algorithms against adaptive edge adversaries. Our two key results for $n$-node graphs of diameter $D$ are as follows:
\begin{itemize}
\item For $t=1$, there is a deterministic algorithm that solves the problem within $\widetilde{O}(D^2)$ rounds\footnote{As usual, the notation $\widetilde{O}(\cdot)$ hides poly-logarithmic terms in the number of nodes $n$.}, provided that the graph is $3$ edge-connected. This round complexity beats the natural barrier of $O(D^3)$ rounds, the existential lower bound on the maximal length of $3$ edge-disjoint paths between a given pair of nodes in $G$. This algorithm can be extended to a $\widetilde{O}(D^{O(t)})$-round algorithm against $t$ adversarial edges in $(2t+1)$ edge-connected graphs.

\item For expander graphs with minimum degree of $\Theta(t^2\log n)$, there is a considerably improved broadcast algorithm with $O(t \log ^2 n)$ rounds against $t$ adversarial edges. This algorithm exploits the connectivity and conductance properties of $G$-subgraphs obtained by employing the Karger's edge sampling technique.	 
	
\end{itemize}
Our algorithms mark a new connection between the areas of fault-tolerant network design and reliable distributed communication. 
\end{abstract}

\newpage
\tableofcontents
\newpage
\setcounter{page}{1}
\section{Introduction}
Guaranteeing the uninterrupted operation of communication networks is a significant objective in network algorithms. The area of resilient distributed computation has been receiving a growing
attention over the last years as computer networks grow in size and become more vulnerable to byzantine failures. 
Since the introduction of this setting by Pease et al. \cite{pease1980reaching} and Lamport et al. \cite{LamportSP82, pease1980reaching} distributed broadcast algorithms against various adversarial models have been studied in theory and practice for over more than four decades. Resilient distributed algorithms have been provided for broadcast and consensus \cite{Dolev82, DolevFFLS82, fischer1983consensus,BrachaT85,toueg1987fast,santoro1989time,BermanGP89,SantoroW90,BermanG93,FeldmanM97,GarayM98,fitzi2000partial,Koo04, PelcP05,katz2006expected,DolevH08, MaurerT12, imbs2015simple,CohenHMOS19,KhanNV19}, as well as for the related fundamental problems of gossiping \cite{blough1993optimal,bagchi1994information,censor2017fast}, and agreement \cite{DolevFFLS82,pease1980reaching,bracha87,CoanW92,GarayM98}. See \cite{pelc1996fault} for a survey on this topic. A key limitation of many of these algorithms is that they assume that the communication graph is the complete graph.

Our paper is concerned with communication graphs of arbitrary topologies. In particular, it addresses the following basic question, which is still fairly open, especially in the \congest\ model of distributed computing \cite{Peleg:2000}:
\begin{question}\label{ques:cost}
	What is the cost (in terms of the number of rounds) for providing resilience against adversarial edges in distributed networks with arbitrary topologies?
\end{question}

An important milestone in this regard was made by Dolev \cite{Dolev82} who showed that the $(2t+1)$ node-connectivity of the graph is a necessary condition for guaranteeing the correctness of the computation in the presence of $t$ adversarial nodes. Pelc \cite{Pelc92} provided the analogous argument for $t$ adversarial edges, requiring edge-connectivity of $(2t+1)$. The broadcast protocol presented therein requires a linear number of rounds (linear in the number of nodes) and exponentially large messages. Byzantine broadcast algorithms for general graph topologies have been addressed mostly under simplified settings \cite{PelcP05}, e.g., probabilistic faulty models \cite{pelc2007feasibility, Pelc92}, cryptographic assumptions \cite{ganesh2016broadcast,AbrahamDDN019,AbrahamCDNP0S19,Nayak0SVX20}, or under bandwidth-\emph{free} settings (e.g., allowing neighbors to exchange exponentially large messages) \cite{Dolev82, MaurerT12, KhanNV19, Koo04, DolevH08,KhanNV19,ChlebusKO20}. 

In this paper, we consider the following extension of the standard \congest\ model to the adversarial setting:
\begin{mdframed}[hidealllines=false,backgroundcolor=white!25]
	\textbf{The Adversarial \congest\ Model:} The network is abstracted as an $n$-node graph $G = (V, E)$, with a
	processor on each node. Each node has a unique identifier of $O(\log n)$ bits. 
	Initially, the processors only know the identifiers of their incident edges, as well as a polynomial estimate on the number of nodes $n$. 
	
	There is a \emph{computationally unbounded} adversary that controls a fixed set of at most $t$ edges in the graph, denoted hereafter as \emph{adversarial} edges. The nodes do not know the identity of the adversarial edges, but they know the bound $t$. The adversary knows the graph topology, the random coins of the nodes, and the protocol executed by the nodes. In each round, it is allowed to send $O(\log n)$ bit messages on each of the adversarial edges $F$ (possibly a distinct message on each edge direction). It is \emph{adaptive} as it can determine its behavior in round $r$ based on the messages exchanged throughout the entire network up to round $r$. 
\end{mdframed}
The definition naturally extends to adversarial \emph{nodes} $F \subseteq V$ for which the adversary can send
in each round, arbitrarily bad $O(\log n)$ bit messages on each of the edges incident to $F$. The primary complexity measure of this model is the \emph{round} complexity. In contrast to many prior works in the adversarial setting, in our model, the nodes are \emph{not} assumed to know the graph's topology, and not even its \emph{diameter}. 
To address \Cref{ques:cost}, we provide a comprehensive study of the adversarial broadcast problem, which is formally defined as follows:
\begin{mdframed}[hidealllines=false,backgroundcolor=white!25]
	\textbf{The adversarial broadcast task:} Given is a $(2t+1)$ edge-connected graph $G=(V,E)$ and a set $F \subset E$ of $|F|\leq t$ edges controlled by the adversary. There is a designated source node $s \in V$ that holds a message $m_0$. It is then required for all the nodes to output $m_0$, while ignoring all other messages. 
\end{mdframed} 
To this date, all existing broadcast algorithms in the adversarial \congest\ model require a polynomial number of rounds, even when handling a single adversarial edge! Recently, Chlebus, Kowalski, and Olkowski \cite{ChlebusKO20} extended the result of Garay and Moses \cite{GarayM98} to general $(2t+1)$ node-connected graphs with minimum degree $3t$. Their algorithms, however, use \emph{exponentially} large communication. Their message size can be improved to polynomial only when using authentication schemes (which we totally avoid in this paper). It is also noteworthy that the existing protocols for \emph{node} failures might still require polynomially many rounds for general graphs, even for a single adversarial edge and for \emph{small} diameter graphs. 

A natural approach for broadcasting a message $m_0$ in the presence of $t$ adversarial edges is to route the message along $(2t+1)$ edge-disjoint paths between the source node $s$, and each target node $v$. This allows each node to deduce $m_0$ by taking the majority message. This approach has been applied in previous broadcast algorithms (e.g., \cite{Dolev82, Pelc92}) under the assumption that the nodes know \emph{the entire graph}, and therefore can compute these edge-disjoint paths. A recent work of \cite{ByzCompilersHP21} demonstrated that there are $D$-diameter 
$(2t+1)$ edge-connected graphs, for which the maximal length of any collection of $(2t+1)$ edge-disjoint paths between a given pair of nodes might be as large as $(D/ t)^{\Theta(t)}$. For $t=1$, the length lower bound becomes $\Omega(D^3)$. Providing round efficient algorithms in the adversarial \congest\ model calls for a new approach.

\paragraph{Our approach in a nutshell.} Our approach is based on combining the perspectives of fault-tolerant (FT) network design, and distributed graph algorithms. The combined power of these points of view allows us to characterize the round complexity of the adversarial broadcast task as a function of the graph diameter $D$, and the number of adversarial edges $t$. This is in contrast to prior algorithms that obtain a polynomial round complexity (in the number of nodes). On a high level, one of the main tools that we borrow from FT network design 
is the FT sampling technique \cite{alon1995color,weimann2013replacement,dinitz2011fault,grandoni2019faster,parter2019small,ParterYPODC19,ChakrabortyC20}, and its recent derandomization by \cite{bodwin2021optimal,RPC21}. For a given graph $G$ and a bound on the number of faults $k$, the FT sampling technique defines a small family $\mathcal{G}=\{G_i \subseteq G\}$ of $G$-subgraphs denoted as \emph{covering family}, which is formally defined as follows:
\begin{definition}[$(L,k)$ covering family]\label{def:covering-Lk}
	For a given graph $G$, a family of $G$-subgraphs $\mathcal{G}=\{G_1,\ldots, G_\ell\}$ is an $(L,k)$ \emph{covering family} if for every $\langle u,v,E' \rangle \in V \times V \times E^{\leq k}$ and any $L$-length $u$-$v$ path $P \subseteq G \setminus E'$, there exists a subgraph $G_i$ such that (P1) $P \subseteq G_i$ and (P2) $E' \cap G_i=\emptyset$.
\end{definition}
As the graph topology is unknown, one cannot hope to compute a family of subgraphs that are completely known to the nodes. 
Instead, we require the nodes to \emph{locally} know the covering family in the following manner. 

\begin{definition}[Local Knowledge of a Subgraph Family]
	A family of ordered subgraphs $\mathcal{G}=\{G_1,\ldots, G_\ell\}$ where each $G_i \subseteq G$, is \emph{locally known} if given the identifier of an edge $e=(u,v)$ and an index $i$, the nodes $u,v$ can locally determine if $e \in G_i$.
\end{definition}
In the context of $(2t+1)$ edge-connected graphs with $t$ adversarial edges, we set $L=O(t D)$ and $k=O(t)$. The randomized FT sampling technique \cite{weimann2013replacement,dinitz2011fault} provides an $(L,k)$ covering family $\mathcal{G}$ of cardinality $O(L^k \log n)$. \cite{RPC21} provided a deterministic construction with $O((L \poly(\log n))^{k+1})$ subgraphs. 

One can show that by the properties of the covering family, exchanging the message $m_0$ over all subgraphs in $\mathcal{G}$ (in the adversarial \congest\ model) guarantees that all nodes successfully receive $m_0$. This holds since for every $v \in V$ and a fixed set of adversarial edges $F$, the family $\mathcal{G}$ contains a subgraph $G_i$ which contains a short $s$-$v$ path (of length $L$) and does not contain any of the adversarial edges. 
Given this observation, our challenge is two folds:
\begin{enumerate}[noitemsep]
	\item provide a round-efficient algorithm for exchanging $m_0$ over all $\mathcal{G}$-subgraphs simultaneously,
	\item guarantee that each node outputs the message $m_0$ while ignoring the remaining messages.
\end{enumerate}

To address the first challenge, we show that the family of subgraphs obtained by this technique has an additional key property of \emph{bounded width}. Informally, a family $\mathcal{G}$ of $G$-subgraphs has a bounded width if each $G$-edge appears in all but a bounded number of subgraphs in $\mathcal{G}$. The bounded width of  $\mathcal{G}$ allows us to exchange messages in all these subgraphs simultaneously, in a nearly optimal number of rounds. The round complexity of this scheme is based on a very careful analysis which constitutes the key technical contribution in this paper. To the best of our knowledge, the bounded width property of the FT sampling technique has been used before only in the context of data structures \cite{weimann2013replacement,0001W20}. It is therefore interesting to see that it finds new applications in the context of reliable distributed communication. The second challenge is addressed by performing an additional communication phase which filters out the corrupted messages. As we will see, the round complexities of our broadcast algorithms for general graphs will be dominated by the cardinality of covering families (which are nearly tight in a wide range of parameters, as shown in \cite{RPC21}).

We also consider the family of expander graphs, which received a lot of attention in the context of distributed resilient computation \cite{DworkPPU88, Upfal94, KingSSV06, AugustineP015}. For these graphs, we are able to show covering families\footnote{Using a somewhat more relaxed definition of these families.} of considerably smaller cardinality that scales linearly with the number of the adversarial edges. This covering family is obtained by using Karger's edge sampling technique \cite{karger1999random}, and its conductance-based analysis by Wulff-Nilsen \cite{wulff2017fully}. We hope this result will also be useful in the context of FT network design. We next describe our key contribution in more detail.

\subsection*{Our Results}
We adopt the gradual approach of fault tolerant graph algorithms, and start by studying broadcast algorithms against a single adversarial edge. Perhaps surprisingly, already this case has been fairly open. We show:
\begin{mdframed}[hidealllines=true,backgroundcolor=gray!25]
	\vspace{-8pt}
	\begin{theorem} \label{lem:broadcast}[Broadcast against a Single Adversarial Edge]
		Given a $D$--diameter, $3$ edge-connected graph $G$, there exists a deterministic algorithm for broadcast against a single adversarial edge that runs in $\widetilde{O}(D^2)$ adversarial-\congest\ rounds. In addition, at the end of the algorithm, all nodes also compute a linear estimate for the \emph{diameter} of the graph.
	\end{theorem}
\end{mdframed}
This improves considerably upon the (implicit) state-of-the-art $n^{O(D)}$ bound obtained by previous algorithms (e.g., by \cite{MaurerT12,ChlebusKO20}). 
In addition, in contrast to many previous works (including \cite{MaurerT12, ChlebusKO20}), our algorithm does not assume global knowledge of the graph or any estimate on the graph's diameter. In fact, at the end of the broadcast algorithm, the nodes also obtain a linear estimate of the graph diameter. 

Using the covering family obtained by the standard FT-sampling technique, it is fairly painless to provide a broadcast algorithm with a round complexity of $\widetilde{O}(D^3)$. Our main efforts are devoted to improving the complexity to $\widetilde{O}(D^2)$ rounds. Note that the round complexity of $D^3$ appears to be a natural barrier for this problem for the following reason. There exists a $3$ edge-connected $D$-diameter graph $G=(V,E)$ and a pair of nodes $s,v$ such that in any collection of $3$ edge-disjoint $s$-$v$ paths $P_1,P_2,P_3$, the length of the longest path is $\Omega(D^3)$ (By Corollary 40 of \cite{ByzCompilersHP21}). The improved bound of $\widetilde{O}(D^2)$ rounds is obtained by exploiting another useful property of the covering families of \cite{RPC21}. One can show that, in our context, each $G$-edge appears on all but $O(\log n)$ many subgraphs in the covering family. This plays a critical role in showing that the simultaneous message exchange on all these subgraphs can be done in $\widetilde{O}(D^2)$ rounds (i.e., linear in the number of subgraphs in this family). 

\paragraph{Multiple adversarial edges.}
We consider the generalization of our algorithms to support $t$ adversarial edges. For $t=O(1)$, we provide broadcast algorithms with $\poly(D)$ rounds. 
\\
\begin{mdframed}[hidealllines=true,backgroundcolor=gray!25]
	\vspace{-8pt}
	\begin{theorem} \label{thm:t-broadcast}[Broadcast against $t$-Adversarial Edges]
		There exists a deterministic broadcast algorithm against $t$ adversarial edges, for every $D$--diameter $(2t+1)$ edge-connected graph, with round complexity of $(tD\log n)^{O(t)}$. Moreover, this algorithm can be implemented in 
		$O(t D \log n)$ \local\ rounds (which is nearly optimal).
	\end{theorem} 
\end{mdframed}
We note that we did not attempt to optimize for the constants in the exponent in our results for multiple adversarial edges. The round complexity of the algorithm is mainly dominated by the number of subgraphs in the covering family. 

\paragraph{Improved broadcast algorithms for expander graphs.}
We then turn to consider the family of expander graphs, which has been shown to have various applications in the context of resilient distributed computation \cite{DworkPPU88, Upfal94, KingSSV06, AugustineP015}. Since the diameter of expander graphs is logarithmic, the algorithm of Theorem \ref{thm:t-broadcast} yields a round complexity of $(t\log n)^{O(t)}$. We provide a considerably improved solution using a combination of tools. The improved broadcast algorithm is designed for $\phi$-expander graphs with minimum degree $\Theta(t^2 \log n/ \phi)$. One can show that these graphs also have a sufficiently large edge-connectivity.

\begin{theorem}\label{thm:expander-broadcast}
	Given an $n$-node $\phi$-expander graph with minimum degree $\Theta(t^2 \log n /\phi)$,
	there exists a randomized broadcast algorithm against $t$ adversarial edges with round complexity of $O(t \cdot \log ^2 n/\phi)$ rounds.	
\end{theorem}
To obtain this result, we employ the edge sampling technique by Karger~\cite{karger1999random}. The correctness arguments are based on Wulff-Nilsen \cite{wulff2017fully}, who provided an analysis of this technique for expander graphs with large minimum degree. Due to the challenges arise by the adversarial setting, the implementation of the sampling technique requires a somewhat larger minimum degree than that required by the original analysis of \cite{wulff2017fully}.

\paragraph{Road Map.} 
The broadcast algorithm against a single adversarial edge and the proof of \Cref{lem:broadcast} are given in \Cref{sec:broadcast-algo}. In \Cref{sec:t-faults} we consider multiple adversarial edges and prove \Cref{thm:t-broadcast}.
In \Cref{sec:expanders}, we show nearly optimal algorithms for expander graphs, providing \Cref{thm:expander-broadcast}.

\paragraph{Preliminaries.} For a subgraph $G' \subseteq G$ and nodes $u,v \in V(G')$, let $\pi(u,v,G')$ be the unique $u$-$v$ shortest path in $G'$ where shortest-path ties are decided in a consistent manner. 
For a path $P=[u_1,\ldots, u_k]$ and an edge $e=(u_k,v)$, let $P \circ e$ 
denote the path obtained by concatenating $e$ to $P$. 
Given a path $P=[u_1,\ldots, u_k]$ denote the sub-path from $u_i$ to $u_\ell$ by $P[u_i,u_{\ell}]$.
The asymptotic term $\widetilde{O}(\cdot)$ hides poly-logarithmic factors in the number of nodes $n$. 
Throughout, we use the following observation. 
\begin{observation}\label{obs:bounded-diam-f}
	Consider an $n$-node $D$-diameter graph $G=(V,E)$ and let $u,v$ be a pair of nodes that are connected in $G \setminus F$ for some $F \subseteq E$. It then holds that $\dist_{G \setminus F}(u,v)\leq 2(|F|+1)\cdot D+|F|$.
\end{observation}
\begin{proof}
	Let $T$ be a BFS tree in $G$ rooted at some source $s$. The forest $T \setminus F$ contains at most $|F|+1$ trees of diameter $2D$. Then, the $u$-$v$ shortest path $P$ in $G \setminus F$ can be transformed into a path $P'$ containing at most $|F|$ edges of $P$, as well as, $|F|+1$ tree subpaths of the forest $T \setminus F$. Therefore, $|P'|\leq 2(|F|+1)\cdot D+|F|$ as desired.
\end{proof}

\section{Broadcast Algorithms against an Adversarial Edge} \label{sec:broadcast-algo}
\vspace{-5pt}
In this section, we prove \Cref{lem:broadcast}. 
We first assume, in \Cref{sec:broadcast-known} that the nodes have a linear estimate $c\cdot D$ on the diameter of the graph $D$, for some constant $c \geq 1$. A-priori, obtaining the diameter estimation seems to be just as hard as the broadcast task itself. In \Cref{sec:notknowing}, we then show how this assumption can be removed. 
Throughout, we assume that the message $m_0$ consists of a single bit. In order to send a $O(\log n)$ bit message, the presented algorithm is repeated for each of these bits (increasing the round complexity by a $O(\log n)$ factor). 

\subsection{Broadcast with a Known Diameter}\label{sec:broadcast-known}
We first describe the adversarial broadcast algorithm assuming that the nodes have a linear estimate $D'$ on the diameter $D$, where $D' \in [D, cD]$ for some constant $c\geq 1$. In \Cref{sec:notknowing}, we omit this assumption. 
The underlying objective of our broadcast algorithm is to exchange messages over reliable communication channels that avoid the adversarial edge $e'$. There are two types of challenges: making sure that all the nodes first \emph{receive} the message $m_0$, and making sure that all nodes correctly \emph{distinguish} between the true bit and the false one. Alg. $\BBalgo(D')$ has two phases, a \emph{flooding} phase and an \emph{acceptance} phase, which at the high level, handles each of these challenges respectively. 

The first phase propagates the messages over an ordered collection of $G$-subgraphs $\mathcal{G}=\{G_1,\ldots, G_\ell\}$ where each $G_i \subseteq G$ has several desired properties. Specifically, $\mathcal{G}$ is an 
$(L,k)$ covering family for $L=O(D')=O(D)$ and $k=1$ (see Def. \ref{def:covering-Lk}). An important parameter of $\mathcal{G}$ which determines the complexity of the algorithm is denoted as the \emph{width}.
\begin{definition}[Width of Covering Family]
	The \emph{width} of a collection of subgraphs $\mathcal{G}=\{G_1,\ldots, G_k\}$, denoted by $\omega(\mathcal{G})$, is the maximal number of subgraphs avoiding a fixed edge in $G$. That is,
	$$\omega(\mathcal{G})=\max_{e \in G}|\{G_i \in \mathcal{G}~\mid~ e\notin G_i\}|~.$$
\end{definition}
The broadcast algorithm starts by applying a $0$-round procedure that provides each node in the graph with a local knowledge of an $(O(D'),1)$ covering family with bounded width. By \cite{RPC21}, we have the following (see \Cref{app:known} for the proof):
\begin{fact}[\cite{RPC21}]\label{fc:rpc}
	Given a $3$ edge-connected graph $G$, there exists a $0$-round algorithm that allows all nodes to locally know an $(L,1)$ covering family $\mathcal{G}=\{G_1,\ldots, G_\ell\}$ for $L=O(D')$, such that $\ell=\widetilde{O}((D')^2)$. The width of $\mathcal{G}$ is $\widetilde{O}(D')$.
\end{fact}
\def\APPENDPRFACT{
	\begin{proof}[Proof of \Cref{fc:rpc}]
		The proof follows by using the construction of $(L,f)$ Replacement Path Cover (RPC) $\mathcal{G}$ by \cite{RPC21} taking $L=O(D')$ and $f=1$. Assuming that node IDs are in $[1,\poly(n)]$, each node can employ this construction locally and compute a collection of subgraphs $\mathcal{G}'=\{G'_1,\ldots, G'_k\}$ where each $G'_i$ is a subgraph of the complete graph $G^*$ on all nodes with IDs in $[1,\poly(n)]$. The set $\mathcal{G}=\{G_1,\ldots, G_k\}$ is defined by $G_i=E(G'_i) \cap E(G)$. We next claim that since  $\mathcal{G}'$ is an $(L,f)$ RPC for $G^*$ it also holds that $\mathcal{G}$ is an $(L,f)$ RPC for $G$. To see this, consider any $L$-length $u$-$v$ path $P$ in $G \setminus \{e\}$. Clearly both $e$ and $P$ are in $G^*$. Therefore there exists a subgraph $G'_i$ containing $P$ and avoiding $e$. It then holds that $G_i =G \cap G'_i$ contains $P$ and avoids $e$ as well. 
		
		As for the width property, the construction of \cite{RPC21} builds the covering family using a hit-miss hash family $H=\{h:[m] \to [q] \}$ where $m$ is the number of edges, and $|H|,q=O(L \poly \log m )$.  Then, for each $h\in H$ and $i\in q$ the subgraph $G_{h,i}$  consists of all edges $e'$ such that $h(e') \neq i$. Therefore, an edge $e$ does not appear only in the subgraphs of the form $G_{h,h(e)}$ and there are $|H|=O(L \poly \log m)$ such subgraphs.

		Finally, the family of graphs $\mathcal{G}$ is locally known since in the construction of \cite{RPC21}, each subgraph is identified with a hash function $h:[m] \to [q]$, in some family of hash functions $\mathcal{H}$, and an index $i \in [q]$. The subgraphs to which an edge $e$ belongs depend only on the value of $h(ID(e))$. Since the family of hash functions $\mathcal{H}$  can be locally computed by each node, we have that given an edge identifier, every node can locally compute the indices of the subgraphs in $\mathcal{G}$ to which $e$ belongs. Note that since the nodes do not know $G$, they also do not know the graph $G_i$ but rather the corresponding graph $G'_i$, and thus we only require them to know the index $i$ of the subgraph. 
	\end{proof}
}

In the following, we present a broadcast algorithm whose time complexity depends on several parameters of the covering family. This will establish the argument assuming that all nodes know a linear bound on the diameter $D' \geq D$. 
\begin{theorem} \label{lem:main}
Given is a $3$ edge-connected graph $G$ of diameter $D$, where all nodes know a constant factor upper bound on $D$, denoted as $D'$. Assuming that the nodes locally know an $(L,1)$ covering family $\mathcal{G}$ for $L=7D'$, there exists a deterministic broadcast algorithm against an adversarial edge with $O(\omega(\mathcal{G})\cdot L+|\mathcal{G}|)$ rounds.
\end{theorem}

\paragraph{Broadcast with a known diameter (Proof of \Cref{lem:main}).}
Given a locally known $(L,1)$ covering family $\mathcal{G}=\{G_1,\ldots, G_\ell\}$ for $L=7D'$, the broadcast algorithm has two phases. The first phase consists of $O(L \cdot \omega(\mathcal{G})+|\mathcal{G}|)$ rounds, and the second phase has $O(L)$ rounds. 

\paragraph{Phase $1$: Flooding phase.} The flooding phase consists of $\ell=|\mathcal{G}|$ sub-algorithms $A_1,\ldots, A_\ell$, where in each algorithm $A_i$, the nodes propagate messages on the underlying subgraph $G_i \in \mathcal{G}$ that is defined locally by the nodes. The algorithm runs the sub-algorithms $A_1,\ldots, A_\ell$ in a pipeline manner, where in the $i$'th round of sub-algorithm $A_i$, the source node $s$ sends the message $(m_0,i)$ to all its neighbors. For every $i\in \{1,\ldots, \ell\}$, a node $u \in V$ that received a message $(m',i)$ from a neighbor $w$, stores the message $(m',i)$ and sends it to all its neighbors if the following conditions hold: (i) $(w,u) \in G_i$, and (ii) $u$ did not receive a message $(m',i)$ in a prior round\footnote{If it receives several $(m',i)$ messages in the same round, it will be considered as only one.}. For a node $u$ and messages $(m_1,i_1), \ldots , (m_k,i_k)$ waiting to be sent in some round $\tau$, $u$ sends the messages according to the order of the iterations $i_1, \ldots , i_k$ (note that potentially $i_j=i_{j+1}$, and there might be at most two messages with index $i_j$, namely, $(0,i_j)$ and $(1,i_j)$).

\paragraph{Phase $2$: Acceptance phase.} The second phase consists of $O(L)$ rounds, in which $accept$ messages are sent from the source $s$ to all the nodes in the graph, as follows. In the first round, the source node $s$ sends an $accept(m_0)$ message to all its neighbors. Then every other node $u \in V$ accepts a message $m'$ as its final output and sends an $accept(m')$ message to all neighbors, provided that the following conditions hold: 
\begin{itemize}
\item[(i)] there exists $i \in \{1,\ldots, \ell\}$, such that $u$ stored a message $(m',i)$ in Phase 1; 
\item[(ii)] $u$ received an $accept(m')$ message in Phase 2 from a neighbor $w_2$, such that $(u,w_2) \notin G_i$. 
\end{itemize}
Since $\mathcal{G}$ is locally known, $u$ can locally verify that $(u,w_2) \notin G_i$.  
This completes the description of the algorithm. 

\paragraph{Correctness.}
We next prove the correctness of the algorithm. 
Denote the adversarial edge by $e'=(v_1,v_2)$.
We begin with showing that no node accepts a wrong message.
\begin{claim} \label{clm:no-false}
	No node $u \in V$ accepts a false message $m' \neq m_0$.
\end{claim}
\begin{proof}
	Assume towards contradiction there exists at least one node which accepts a message $m' \neq m_0$ during Phase 2. Let $u$ be the \emph{first} node that accepts $m'$. By first we mean that any other node that accepted $m'$, accepted the message in a later round than $u$, breaking ties arbitrarily.
	Hence, by Phase 2, $u$ received an $accept(m')$ message from a neighbor $w$, and stored a message $(m',i)$ in Phase 1, where $(u,w) \notin G_i$. Since $u$ is the first node that accepts $m'$, the node $w$ did not accept $m'$ in the previous round. We conclude that the edge $(w,u)$ is the adversarial edge, and all other edges are reliable. Because the adversarial edge $(w,u)$ was not included in the $i$'th graph $G_i$, all messages of the form $(m',i)$ sent by the adversarial edge in Phase 1 are ignored. Since all other edges are reliable, all nodes ignored the false message $(m',i)$ during the first phase (in case they recieved it), in contradiction to the assumption that $u$ stored $(m',i)$ in Phase 1.
\end{proof}
From \Cref{clm:no-false} we can also deduce that in the case where the adversarial edge initiates a false broadcast, it will not be accepted by any of the nodes.
\begin{corollary} \label{cor:no-accept}
In case $e'=(v_1,v_2)$ initiates the broadcast, no node in $G$ will accept any message.
\end{corollary}
\begin{proof}
Since no node initiated the broadcast, in the second phase the only nodes that can receive $accept(m)$ messages are $v_1$ and $v_2$ over the adversarial edge $e'$. In addition, since $e'$ also initiates the first phase, for every node storing a message $(m,i)$ in Phase 1 it must hold that $e'\in G_i$. Hence, we can conclude that neither $v_1$ nor $v_2$ accepts any of the false messages. Consequently, no node in $V \setminus \{v_1,v_2\}$ receives an $accepts(m)$ message for any $m$, as required.
\end{proof}
So far, we have shown that if a node $v$ accepts a message, it must be the correct one. It remains to show that each node indeed accepts a message during the second phase. 
Towards that goal, we will show that the collection of $\ell$ sub-algorithms executed in Phase 1 can be simulated in $O(\omega(\mathcal{G}) \cdot L+|\mathcal{G}|)$ rounds. This argument holds regardless of the power of the adversary.
\begin{lemma} \label{lem:pipeline}
Consider an $(L,1)$ covering family $\mathcal{G}=\{G_1,\ldots, G_\ell\}$ for $G$ that is locally known by all the nodes. For a fixed node $v$, an edge $e$, and an $L$-length $s$-$v$ path $P \subseteq G \setminus \{e,e'\}$, let $G_{i} \in \mathcal{G}$ be the subgraph containing $P$ where $e \notin G_i$. Then, $v$ receives the message $(m_0,i)$ in Phase 1 within $O(L \cdot \omega(\mathcal{G})+|\mathcal{G}|)$ rounds. 
	
\end{lemma}
We note that for $D' \geq D$, by \Cref{obs:bounded-diam-f} taking $L=7D'$ yields that for every node $v$ and edge $e$, it holds that $\dist_{G \setminus \{e,e'\}}(s,v) \leq L$. Hence, by the properties of the covering family $\mathcal{G}$ (Def. \ref{def:covering-Lk}), for every node $v$ and an edge $e$ there exists an $L$-length $s$-$v$ path $P\subseteq G \setminus \{e,e'\}$ and a subgraph $G_i$ that contains $P$ and avoids $e$. The proof of \Cref{lem:pipeline} is one of the most technical parts in this paper. Whereas pipeline is a very common technique, especially in the context of broadcast algorithms, our implementation of it is quite nontrivial. Unfortunately, since our adversary has full knowledge of the randomness of the nodes, it is unclear how to apply the random delay approach of ~\cite{Ghaffari15,leighton1994packet} in our setting. We next show that our pipeline approach works well due to the bounded width of the covering family.

\paragraph{Proof of \Cref{lem:pipeline}.}
Let $P=(s=v_0, \ldots ,v_\eta=v)$ be an $s$-$v$ path in $G_i \setminus \{e'\}$, where $\eta \leq L$ and $e \notin G_i$. 
For simplicity, we consider the case where the only message propagated during the phase is $m_0$. The general case introduces a factor of 2 in the round complexity. This holds since there could be at most two messages of the form $(0, i)$ and $(1, i)$. We also assume, without loss of generality, that each node $v_j$ receives the message $(m_0, i)$ for the \emph{first} time from $v_{j-1}$. If $v_j$ received $(m_0, i)$ for the first time from a different neighbor in an earlier round, the time it sends the message can only decrease.

In order to show that $v_{\eta}=v$ receives the message $(m_0,i)$ within $O(L \cdot \omega(\mathcal{G})+|\mathcal{G}|)$ rounds, it is enough to bound the total number of rounds the message $(m_0,i)$ spent in the \emph{queues} of the nodes of $P$, waiting to be sent. That is, for every node $v_j \in P$, let $r_j$ be the round in which $v_j$ received the message $(m_0,i)$ for the first time, and let $s_j$ be the round in which $v_j$ sent the message $(m_0,i)$. 
In order to prove \Cref{lem:pipeline}, our goal is to bound the quantity 
\begin{align}
	T = \sum_{j=1}^{\eta-1}(s_j-r_j). \label{eq:def-bound}
\end{align}
For every $k<i$ we denote the set of edges from $P$ that are not included in the subgraph $G_k$ by $$N_k=\{(v_{j-1},v_j) \in P~\mid~ (v_{j-1},v_j) \notin G_k\},$$ and define $\mathcal{N} = \{(k,e) ~\mid~ e \in N_k, k \in \{1,\ldots, i-1\}\}.$ By the definition of the width property, it holds that: 
\begin{align}
	|\mathcal{N}|=\sum_{k=1}^{i-1} |N_k| \leq \eta \cdot \omega(\mathcal{G}) =O(\omega(\mathcal{G}) \cdot L). \label{eq:boundN}
\end{align}
By \Cref{eq:def-bound,eq:boundN} it follows that in order to prove \Cref{lem:pipeline} its enough to show $T \leq |\mathcal{N}|$.
For every node $v_j \in P$, let $Q_j$ be the set of messages $(m_0,k)$ that $v_j$ sent between rounds $r_j$ and round $s_j$. By definition, $|Q_j|=(s_j-r_j)$, and  $T= \sum_{j=1}^{\eta-1}|Q_j|$. 
We next show that
$\sum_{j=1}^{\eta-1}|Q_j| \leq |\mathcal{N}|.$ This is shown in two steps. First we define a set $I_j$ consisting of certain $(m_0,k)$ messages such that $|Q_j| \leq |I_j|$, for every $j \in \{1,\ldots, \eta-1\}$. Then, we show that $\sum_{j=1}^{\eta-1}|I_j| \leq |\mathcal{N}|$. 

\paragraph{Step one.}
For every node $v_j\in P$, let $I_j$ be the set of messages $(m_0,k)$ satisfying the following three properties : (1) $k < i$, (2) $v_j$ sent the message $(m_0,k)$ before sending the message $(m_0,i)$ in round $s_j$, and (3) $v_j$ did not receive the message $(m_0,k)$ from $v_{j-1}$ before receiving the message $(m_0,i)$ in round $r_j$. In other words, the set $I_j$ includes messages received by $v_j$  with a graph index at most $(i-1)$, that are either received from $v_{j-1}$ between round $r_j$ and round $s_j$, or received by $v_j$ from another neighbor $w \neq v_{j-1}$ by round $s_j$ (provided that those messages were not received additionally from $v_{j-1}$).  Note that it is not necessarily the case that $Q_j \subseteq I_j$, but for our purposes, it is sufficient to show the following.
\begin{claim} \label{clm:reduction}
	For every $1 \leq j \leq \eta-1$ it holds that $|Q_j| \leq |I_j|$.
\end{claim}
\begin{proof}
	We divide the set $Q_j$ into two disjoint sets, $Q_{j,1}$ and $Q_{j,2}$. Let $Q_{j,1}$ be the set of messages $(m_0,k)$ in $Q_j$ that $v_j$ received from $v_{j-1}$ \emph{before} round $r_j$, and let $Q_{j,2}=Q_j \setminus Q_{j,1}$, i.e., the set of messages in $Q_j$ that $v_j$ did not receive from $v_{j-1}$ by round $r_j$. 
	We first note that by the definition of the set $I_j$, it holds that $Q_{j,2} \subseteq I_j$.
	Second, since $v_j$ did not send the messages in $Q_{j,1}$ before round $r_j$, and the messages are sent in a pipeline manner according to the subgraph index $k$, there exists $|Q_{j,1}|$ messages denoted as $I'$ with index at most $(i-1)$ (i.e., message of the form $(m_0,k)$ for $k\leq (i-1)$), that $v_j$ did not receive from $v_{j-1}$, and were sent before round $r_j$. Hence, it holds that $I' \subseteq I_j \setminus Q_j$, and therefore $|I_j \setminus Q_{j,2}| \geq |Q_{j,1}|$. The claim follows.
\end{proof}
\paragraph{Step two.}
We next show that $\sum_{j=1}^{\eta-1}|I_j| \leq |\mathcal{N}|$ by introducing an injection function $f$ from $\mathcal{I}= \{(v_j,k) ~\mid~ (m_0,k) \in I_j\}$
to $\mathcal{N}$, defined as follows.
For $(v_j,k) \in \mathcal{I}$, we set $f((v_j,k))=(k, (v_{h-1},v_h))$ such that $(v_{h-1},v_h)$ is the closest edge to $v_j$ on $P[v_0,v_j]$, where $(v_{h-1},v_h)\in N_k$ (i.e., $(v_{h-1},v_h) \notin G_k$). 
Specifically,
\begin{equation}
	f((v_j,k))=(k, (v_{h-1},v_h)) ~\mid~ h=\max_{\tau\leq j} \{\tau ~\mid~ (v_{\tau-1},v_{\tau})\in N_k \}~. \label{eq:function}
\end{equation}
We begin by showing that the function is well defined.
\begin{claim} \label{clm:well-defined}
	The function $f:\mathcal{I} \rightarrow \mathcal{N}$ is well defined. 
\end{claim}
\begin{proof}
	For a pair $(v_j,k) \in \mathcal{I}$, we will show that there exists an edge in the path $P[v_0,v_j]$ that is not included in $G_k$, and therefore $f((v_j,k))$ is defined.
Assume by contradiction that all path edges in $P[v_0,v_j]$ are included in $G_k$. We will show by induction on $h$ that for every $1 \leq h \leq j$ the node $v_h$ receives the message $(m_0,k)$ from $v_{h-1}$ before receiving the message $(m_0,i)$ from $v_{h-1}$, leading to a contradiction with property (3) in the definition of $I_j$, as $(v_j,k)\in I_j$.
	
	For the base case of $h=1$, at the beginning of the flooding procedure the source $s=v_0$ sends the messages $(m_0,1), \ldots, (m_0,i)$ one after the other in rounds $1,\ldots, i$ respectively.  Since the edge $(v_0,v_1) \in G_k$, the node $v_1$ receives the message $(m_0,k)$ from $v_{0}$ in round $k$, and in particular before receiving the message $(m_0,i)$ (in round $i$).  Assume the claim holds up to node $v_{h}$, and we will next show correctness for $v_{h+1}$. By the induction assumption, $v_h$ receives the message $(m_0,k)$ from $v_{h-1}$ before it receives the message $(m_0,i)$. 
	Hence, $v_h$ also sends $(m_0,k)$ before sending $(m_0,i)$.  Because $(v_{h},v_{h+1}) \in G_k$, the message $(m_0,k)$ is not ignored, and $v_{h+1}$ receives the message $(m_0,k)$ from $v_{h}$ before receiving the message $(m_0,i)$. Note that this proof heavily exploits the fact that all the edges on $P$ are reliable.
\end{proof}

Next, we show that the function $f$ is an injection.
\begin{claim}
	The function $f$ is an injection. 
\end{claim}
\begin{proof}
	First note that by the definition of the function $f$ (see Eq. \ref{eq:function}), for every $k_1 \neq k_2$, and $1 \leq j_1,j_2 \leq \eta-1$ such that $(v_{j_1},k_1),(v_{j_2},k_2) \in \mathcal{I}$, it holds that $f((v_{j_1},k_1)) \neq f((v_{j_2},k_2))$. 
	Next, we show that for every $k < i$ and $1 \leq j_1 < j_2 \leq \eta-1$ such that $(v_{j_1},k),(v_{j_2},k) \in \mathcal{I}$, it holds that $f((v_{j_1},k)) \neq f((v_{j_2},k))$. Denote by $f((v_{j_1},k)) = (k,(v_{h_1-1},v_{h_1}))$ and $f((v_{j_2},k)) = (k,(v_{h_2-1},v_{h_2}))$. We will now show that $(v_{h_2-1},v_{h_2}) \in P[v_{j_1}, v_{j_2}]$. Since $(v_{h_1-1},v_{h_1}) \in P[v_0,v_{j_1}]$, it will then follow that $(v_{h_1-1},v_{h_1}) \neq (v_{h_2-1},v_{h_2})$.  
	
	Assume towards contradiction that $(v_{h_2-1},v_{h_2}) \in P[v_0,v_{j_1}]$. By the definition of $f$ and the maximality of $h_2$, it holds that $P[v_{j_1},v_{j_2}] \subseteq G_k$. Since $(v_{j_1},k) \in \mathcal{I}$, the node $v_{j_1}$ sent the message $(m_0,k)$ before sending $(m_0,i)$. Since we assumed every node $v_t \in P$ receives the message $(m_0,i)$ for the first time from $v_{t-1}$ (its incoming neighbor on the path $P$), and all the edges on $P$ are reliable, it follows that $v_{j_2}$ received the message $(m_0,k)$ from $v_{j_2-1}$ \emph{before} receiving the message $(m_0,i)$ in round $r_{j_2}$. This contradicts the assumption that $(v_{j_2},k) \in \mathcal{I}$, as by property (3) in the definition of $I_{j_2}$, the node $v_{j_2}$ did not receive the message $(m_0,k)$ from $v_{j_2-1}$ before round $r_{j_2}$.
\end{proof}
This completes the proof of \Cref{lem:pipeline}.
Finally, we show that when $D' \geq D$, all nodes accept the message $m_0$ during the second phase using \Cref{lem:pipeline}. This will conclude the proof of \Cref{lem:main}. 
\begin{claim} \label{clm:all-accepts}
 All nodes accept $m_0$ within $O(L)$ rounds from the beginning of the second phase, provided that $D' \geq D$.
\end{claim}
\begin{proof}
	Let $T$ be a BFS tree rooted at $s$ in $G \setminus \{e'\}$. For a node $u$, let $p(u)$ be the parent of $u$ in the tree $T$.
	We begin with showing that each node $u$ receives and stores a message $(m_0,j)$ such that $(u,p(u))\notin G_j$ during the first phase. Because the graph is $3$ edge-connected, by \Cref{obs:bounded-diam-f} for every node $u$ there exists an $s$-$v$ path $P$ in $G$ that does not contain both the edge $(u, p(u))$ and $e'$, of length $|P| \leq 7D \leq 7D'=L$.
	By the definition of the covering family $\mathcal{G}$, there exists a subgraph $G_j$ containing all the edges in $P$, and in addition, $(u,p(u)) \notin G_j$. Hence, by \Cref{lem:pipeline}, $u$ stores a message of the form $(m_0,j)$ during the first phase.
	
	We next show by induction on $i$ that all nodes in layer $i$ of the tree $T$ accepts $m_0$ by round $i$ of Phase 2. For the base case of $i=1$, given a node $u$ in the first layer of $T$, $u$ receives the message $accept(m_0)$ from $s$ after one round of Phase 2. Since $u$ stored a message $(m_0,j)$ such that $(s,u)\notin G_j$ in Phase 1, it accepts the message $m_0$ within one round. Assume the claim holds for all nodes up to layer $(i-1)$ and let $u$ be a node in the $i$'th layer. By the induction assumption, $p(u)$ sent $accept(m_0)$ to $u$ by round $(i-1)$. Because $u$ stored a message $(m_0,j)$ such that $(p(u),u) \notin G_j$ during the first phase, $u$ accepts $m_0$ by round $i$.
\end{proof}

\noindent\textbf{Remark.}
Our broadcast algorithm does not need to assume that the nodes know the identity of the source node $s$. By~\Cref{cor:no-accept}, in case where the adversarial edge $e'$ initiates a false broadcast execution, no node will accept any of the messages sent.

\subsection{Broadcast without Knowing the Diameter}\label{sec:notknowing}
We next show how to remove the assumption that the nodes know an estimate of the diameter; consequently, our broadcast algorithm also computes a linear estimate of the diameter of the graph. This increases the round complexity by a logarithmic factor and establishes \Cref{lem:broadcast}. \\
We first describe the algorithm under the simultaneous wake-up assumption and then explain how to remove it. 
\vspace{-5pt}
\paragraph{Algorithm $\UnBBalgo$.}
The algorithm applies Alg. $\BBalgo$ of \Cref{sec:broadcast-algo} for $k=O(\log D)$ iterations in the following manner. Every iteration $i \in \{1,\ldots, k\}$ consists of two steps. In the first step, the source node $s$ initiates Alg. $\BBalgo(D_i)$ with diameter estimate $D_i=2^{i}$, and the desired message $m_0$. 
Denote all nodes that accepted the message $m_0$ by $A_i$ and let $N_i=V \setminus A_i$ be the nodes that did not accept the message. 

In the second step, the nodes in $N_i$ inform $s$ that the computation is not yet completed in the following manner. All nodes in $N_i$ broadcast \emph{the same} designated message $M$ by applying Alg. $\BBalgo(9D_i)$ with diameter estimate $9D_i$ and the message $M$. The second step can be viewed as performing a single broadcast from $|N_i|$ multiple sources. If the source node $s$ receives and accepts the message $M$ during the second step, it continues to the next iteration $(i+1)$. If after $\widetilde{O}(D_i^2)$ rounds $s$ did not receive and accept the message $M$, it broadcasts a termination message $M_T$ to all nodes in $V$ using Alg. $\BBalgo(28D_i)$ (with diameter estimate $28D_i$). Once a node $v \in V$ accepts the termination message $M_T$, it completes the execution with the output message it has accepted so far. Additionally, for an iteration $i$ in which $v$ accepted the termination message, $D_i$ can be considered as an estimation of the graph diameter. 

 \paragraph{Analysis.}
We begin with noting that no node $v \in V$ accepts a wrong message $m' \neq m_0$ as its output. This follows by \Cref{clm:no-false} and the correctness of Alg. $\BBalgo$. 
\begin{observation} \label{clm:corrent}
	No node $v \in V$ accepts a wrong message $m' \neq m_0$.
\end{observation}
Fix an iteration $i$. Our next goal is to show that if $N_i \neq \emptyset$, then $s$ will accept the message $M$ by the end of the iteration. 
Consider the second step of the algorithm, where the nodes in $N_i$ broadcast the message $M$ towards $s$ using Alg. $\BBalgo(9D_i)$. 
Since all nodes in $N_i$ broadcast \emph{the same} message $M$, we refer to the second step as a \emph{single} execution of Alg. $\BBalgo(9D_i)$ with multiple sources.
We begin with showing that the distance between the nodes in $A_i$ and $s$ is at most $14 D_i$. 
\begin{claim}~\label{clm:small-tree}
	For every node $v \in A_i$ it holds that $\dist(s,v,G \setminus \{e'\}) \leq 14D_i$.
\end{claim}
\begin{proof}
	Recall that Alg. $\BBalgo$ proceeds in two phases. In the first phase, the source node propagates messages of the form $(M,k)$, and in the second phase, the source node propagates \emph{accept} messages.  
	For a node $v\in A_i$ that accepts the message $m_0$ in the $i$'th iteration, by the second phase of Alg. $\BBalgo$, it received an $accept(m_0)$ message from a neighbor $w$ in Phase 2, and stored a message $(m_0,k)$  in Phase 1, such that $(v,w) \notin G_k$. 
	Let $P_1$ be the path on which the message $accept(m_0)$ propagated towards $v$ in Phase 2 of Alg. $\BBalgo(D_i)$. Since the second phase is executed for $7D_i$ rounds, it holds that $|P_1| \leq 7 D_i$. In the case where $e' \notin P_1$, since $s$ is the only node initiating $accept(m_0)$ messages (except maybe $e'$), $P_1$ is a path from $s$ to $v$ in $G \setminus \{e'\}$ as required. 
	
	Assume that $e' \in P_1$, and denote it by $e'=(v_1,v_2)$. Without loss of generality, assume that on the path $P_1$,  the node $v_1$ is closer to $v$ than $v_2$. Hence, $v_1$ received an $accept(m_0)$ message from $v_2$ during Phase 2, and because $v_1$ also sent the message over $P_1$, it accepted $m_0$ as its output. Therefore, during the execution of Alg. $\BBalgo(D_i)$, the node $v_1$ stored a message $(m_0, j)$ during the first phase, where $e' \notin G_j$.  As all edges in $G_j$ are reliable, we conclude that $G_j$ contains an $s$-$v_1$ path $P_2$ of length $\eta \leq 7D_i$ such that $e' \notin P_2$. Thus, the concatenated path $P_2 \circ P_1[v_1,v]$ is a path of length at most $14D_i$ from $s$ to $v$ in $G \setminus \{e'\}$ as required. 
\end{proof}
We now show that if $N_i \neq \emptyset$ then $s$ accepts the message $M$ during the second step and continues to the next iteration.
The proof is very similar to the proof of \Cref{clm:all-accepts} and follows from the following observation.
\begin{observation}\label{obs:small-path}
For every $u\in A_i$ and an edge $e=(v,u)$, it holds that $\dist_{G \setminus \{e',e\}}(N_i,u) \leq 7 \cdot 9 D_i$.
\end{observation}
\begin{proof}
Let $T$ be a truncated BFS tree rooted at $s$ in $G \setminus \{e'\}$, such that (1) $A_i \subseteq V(T)$, and (2) the leaf nodes of $T$ are in the $A_i$ set. Informally, $T$ is a minimum depth tree rooted at $s$ in $G \setminus \{e'\}$, that spans the vertices in $A_i$.
By \Cref{clm:small-tree} the depth of $T$ is at most $14 D_i$.
In follows that the forest $T \setminus \{e\}$ contains at most two trees of diameter $2\cdot 14 D_i$. 
Since $G$ is 3 edge-connected, there exists a path from some node in $N_i$ to $u$ in $G \setminus \{e,e'\}$.

Hence, the shortest path from $N_i$ to $u$ in $G \setminus \{e,e'\}$ denoted as $P$ can be transformed into a path $P'$ containing at most one edge in $P$, and two tree subpaths of the forest $T \setminus \{e\}$. Therefore, $\dist_{G \setminus \{e,e'\}}(N_i,u) \leq |P'|\leq 4\cdot 14 D_i+1 \leq 7\cdot 9 D_i$.
\end{proof}
\begin{claim} \label{clm:accept-diameter}
If $N_i \neq \emptyset$, $s$ accepts the message $M$ by the end of Step 2 of the $i$'th iteration. 
\end{claim}
\begin{proof}
Let $P=(u_0,u_1, \ldots ,u_\eta=s)$ be a shortest path from some node $u_0 \in N_i$ to the source node $s$ in $G\setminus\{e'\}$. 
As $P$ is the shortest such path, for every $j \neq 0$ $u_j\in A_i$, and by to \Cref{clm:small-tree} $|P|\leq (14D_i+1)$. 
In order to prove \Cref{clm:accept-diameter}, we will show by induction on $j$ that every $u_j\in P$ accepts the message $M$ by round $j$ of Phase 2 in Alg. $\BBalgo(9D_i)$ (in Step 2). 

For the base case of $j=0$, as $u_0 \in N_i$, it accepts the message $M$ at the beginning of the phase. Assume the claim holds for $u_j$ and consider the node $u_{j+1}$.  By \Cref{obs:small-path} there exists a path $P_{j+1}$ from some node in $N_i$ to $u_{j+1}$ in $G \setminus \{e',(u_j,u_{j+1})\}$ of length $|P_{j+1}|\leq 7\cdot 9D_i$. 
Hence, by \Cref{lem:pipeline} combined with the covering family used in Alg. $\BBalgo(9D_i)$, we conclude that $u_{j+1}$ stored a message $(M,k)$ in Phase 1 where $(u_j,u_{j+1})\notin G_{k}$.
By the induction assumption, $u_j$ sends $u_{j+1}$ an $accept(M)$ message by round $j$ of Phase 2. Since $(u_{j+1},u_j) \notin G_{k}$, it follow that $u_{j+1}$ accepts the message $M$ by round $(j+1)$.
\end{proof}

By \Cref{clm:small-tree} it follows that when $D_i < (D/28)$ there must exist a node $w \in V$ that did not accept the message $m_0$ during the execution of Alg.  $\BBalgo(D_i)$ in the first step, and $N_i \neq \emptyset$.  
On the other hand, when $D_i\geq D$ by \Cref{clm:all-accepts} all nodes accept the message $m_0$ during the first step of the $i$'th iteration, and therefore $N_i = \emptyset$. Hence, for an iteration $i^*$ in which no node broadcasts the message $M$ (and therefore $s$ decides to terminate the execution), it holds that $D_{i^*} \in [D/28,2D]$. Since $s$ broadcasts the termination message $M_T$ by applying Alg. $\BBalgo(28D_{i^*})$ with diameter estimate $28D_{i^*}$, we conclude that all nodes in $V$ will finish the execution as required. 

\smallskip

\noindent \textbf{Omitting the simultaneous wake-up assumption.} The main adaptation is that in the second step of each iteration, the message $M$ is initiated by the nodes in $N_i$ with neighbors in $A_i$ (and possibly also the endpoints of the adversarial edge), rather than all the nodes in $N_i$. Specifically, the modifications are as follows.
At the end of the first step, every node $u\in A_i$ informs all its neighbors that the first step has ended. Every node $u\in N_i$ receiving such a message, initiates the second step by broadcasting the message $M$ using Alg. $\BBalgo(9D_i)$. In the case where some node $u\in N_i$ receives messages indicating that the first step has ended from two distinct neighbors $w,w'$ in two distinct rounds $\tau \neq \tau'$, the node $u$ broadcasts the message $M$ in rounds $(\tau+1)$ and $(\tau'+1)$. We note that this case can occur when at least one of the endpoints of the adversarial edge $e'$ is in $N_i$. A node $u\in A_i$ that receives $M$ during the first step ignores these messages.

The correctness follows by the fact that for every edge $(u,w)\in G \setminus \{e'\}$ such that $u\in N_i$ and $w\in A_i$, the node $u$ broadcasts the message $M$ at the beginning of the second step of iteration $i$. 
For that reason, \Cref{clm:accept-diameter} works in the same manner. Additionally, in an iteration $i^*$ where $N_{i^*}=\emptyset$, no node broadcasts the message $M$ in the second step, and therefore $s$ broadcasts the termination message $M_T$.

\section{Broadcast against $t$ Adversarial Edges} \label{sec:t-faults}
In this section, we consider the broadcast problem against $t$ adversarial edges and prove \Cref{thm:t-broadcast}. The adversarial edges are fixed throughout the execution but are unknown to any of the nodes. Given a $D$--diameter, $(2t+1)$ edge-connected graph $G$, and at most $t$ adversarial edges $F \subseteq E$, the goal is for a source node $s$ to deliver a message $m_0$ to all nodes in the graph. At the end of the algorithm, each node is required to output the message $m_0$. Our algorithm is again based on a locally known covering family $\mathcal{G}$ with several desired properties. The algorithm floods the messages over the subgraphs of $\mathcal{G}$. The messages exchanged over each subgraph $G_i \in \mathcal{G}$ also contains the path information along which the message has been received. As we will see, the round complexity of the algorithm is mostly dominated by the cardinality of $\mathcal{G}$. 
We use the following fact from~\cite{RPC21}, whose proof follows by the proof of \Cref{fc:rpc}.  
\begin{fact}[Implicit in \cite{RPC21}]\label{fc:rpc-t}
	Given a graph $G$ and integer parameters $L$ and $k$, there exists a $0$-round algorithm that allows all nodes to locally know an $(L,k)$ covering family $\mathcal{G}=\{G_1,\ldots, G_\ell\}$ such that $\ell=((Lk\log n)^{k+1})$. 
\end{fact}
Towards proving \Cref{thm:t-broadcast}, we prove the following theorem which will become useful also for the improved algorithms for expander graphs in Sec. \ref{sec:expanders}. 

\begin{theorem} \label{lem:main-t}
	Given is a $(2t+1)$ edge-connected graph $G$ of diameter $D$, and a parameter $L$ satisfying that for every $u,v \in V$, and every set $E' \subseteq E$ of size $|E'| \leq 2t$, it holds that $\dist_{G \setminus E'}(u,v)\leq L$.
	Assuming that the nodes locally know an $(L,2t)$ covering family $\mathcal{G}$, there exists a deterministic broadcast algorithm $\BBtalgo(\cG,L,t)$ against $t$ adversarial edges $F$ with round complexity $O(L \cdot |\mathcal{G}|)$.
\end{theorem}
We note that by \Cref{obs:bounded-diam-f}, every $(2t+1)$ edge-connected graph $G$ with diameter $D$ satisfies the promise of \Cref{lem:main-t} for $L=(6t+2)D$. 
Our algorithm makes use of the following definition for a minimum $s$-$v$ cut defined over a collection of $s$-$v$ \emph{paths}. 
\begin{definition}[Minimum (Edge) Cut of a Path Collection]
	Given a collection of $s$-$v$ paths $\mathcal{P}$, the minimum $s$-$v$ cut in $\mathcal{P}$, denoted as $\mincut(s,v,\mathcal{P})$, is the minimal number of edges appearing on all the paths in $\mathcal{P}$. I.e., letting $\mincut(s,v, \mathcal{P})=x$ implies that there exists a collection of $x$ edges $E'$ such that for every path $P \in \mathcal{P}$, it holds that $E' \cap P\neq \emptyset$. 
\end{definition}
We are now ready to describe the broadcast algorithm given that the nodes know an $(L,2t)$ covering family $\mathcal{G}$ (along with the parameters $L$ and $t$) as specified by Theorem \ref{lem:main-t}. Later, we explain the general algorithm that omits this assumption.
%
%

\paragraph{Broadcast Algorithm $\BBtalgo(\cG,L,t)$.} 
Similarly to the single adversarial edge case, the algorithm has two phases, a flooding phase, and an acceptance phase. In the first phase of the algorithm, the nodes exchange messages over the subgraphs of $\mathcal{G}$, which also contains the path information along which the messages are received. In addition, instead of propagating the messages of distinct $G_i$ subgraphs in a pipeline manner, we run the entire $i$'th algorithm (over the edges of the graph $G_i$) after finishing the application of the $(i-1)$'th algorithm\footnote{One might optimize the $O(t)$ exponent by employing a pipeline approach in this case as well.}. 
\\
In the first phase, the nodes flood \emph{heard bundles} over all the $G_i \in \mathcal{G}$ subgraphs, defined as follows.
\\
\noindent\textbf{Heard bundles}: A \emph{bundle} of heard messages sent from a node $v$ to $u$ consists of: 
\vspace{-4pt}
\begin{enumerate} 
	\item A header message $heard(m,len,P)$, where $P$ is an $s$-$v$ path of length $len$ along which $v$ received the message $m$.
	\vspace{-4pt}
	\item A sequence of $len$ messages specifying the edges of $P$, one by one.
\end{enumerate}
\vspace{-4pt}
This bundle contains $(len+1)$ messages that will be sent in a pipeline manner in the following way. The first message is the header $heard(m,len,P)$. Then in the next consecutive $len$ rounds, $v$ sends the edges of $P$ in reverse order (from the edge incident to $v$ to $s$). 
\vspace{-7pt}
\paragraph{Phase 1: Flooding.}
The first phase consists of $\ell=|\mathcal{G}|$ iterations, where each iteration is implemented using $O(L)$ rounds. 
In the first round of the $i$'th iteration, the source node $s$ sends the message $heard(m_0,1,\emptyset)$ to all neighbors. 
Every node $v$, upon receiving the \emph{first} bundle message $heard(m',x,P)$ over an edge in $G_i$ from a neighbor $w$, stores the bundle $heard(m', x+1 , P\cup \{w\})$ and sends it to all neighbors.
Note that each node stores and sends at most \emph{one} heard bundle $heard(m', x ,P)$ in each iteration (and not one per message $m'$). 
\vspace{-7pt}
\paragraph{Phase 2: Acceptance.}
The second phase consists of $O(L)$ rounds, in which \emph{accept} messages are propagated from the source $s$ to all nodes as follows. In the first round, $s$ sends $accept(m_0)$ to all neighbors. Every node $v \in V \setminus \{s\}$ decides to accept a message $m'$ if the following two conditions hold: (i) $v$ receives $accept(m')$ from a neighbor $w$, and
(ii) $\mincut(s,v, \mathcal{P}) \geq t$, where 
\begin{equation}\label{eq:path-col}
	\mathcal{P}=\{P ~\mid~ \text{$v$ stored a $heard(m',len,P)$  message and } (v,w) \notin P\}~.
\end{equation}
Note that since the decision here is made by computing the minimum cut of a path collection, it is indeed required (by this algorithm) to send the path information.

\vspace{-5pt}
\paragraph{Correctness.}
We begin with showing that no node accepts a false message.
\begin{claim} \label{clm:t-no-false}
	No node $v \in V$ accepts a message $m' \neq m_0$ in the second phase.
\end{claim}
\begin{proof} \vspace{-5pt}
	Assume by contradiction there exists a node that accepts a false message $m'$, and let $v$ be the \emph{first} such node. By first we mean that any other node that accepted $m'$ accepted the message in a later round than $v$, breaking ties arbitrarily.
	Hence, $v$ received a message $accept(m')$ from some neighbor $w$. Because $v$ is the first such node, the edge $(w,v)$ is adversarial. Let $E'= F \setminus \{(w,v)\}$ be the set of the remaining $(t-1)$ adversarial edges, and let $\mathcal{P}$ be given as in Eq. (\ref{eq:path-col}).
	We next claim that $\mincut(s,v, \mathcal{P}) \leq (t-1)$ and therefore $v$ does not accept $m'$. 
	
	To see this, observe that any path $P$ such that $v$ received a message $heard(m', len, P)$ must contain at least one edge in $E'$. This holds even if the content of the path $P$ is corrupted by the adversarial edges. Since there are at most $(t-1)$ edges in $E'$, and each of the paths in $\mathcal{P}$ intersects these edges, it holds that $\mincut(s,v, \mathcal{P}) \leq (t-1)$ as required. 
\end{proof}
Finally, we show that all nodes in $V$ accept the message $m_0$ during the second phase. This completes the proof of \Cref{lem:main-t}.
\begin{claim} \label{clm:t-all-accepts}
	All nodes accept $m_0$ within $O(L)$ rounds from the beginning of the second phase.
\end{claim}
\begin{proof} \vspace{-5pt}
	We will show that all nodes accept the message $m_0$ by induction on the distance from the source $s$ in the graph $G \setminus F$. 
	The base case holds vacuously, as $s$ accepts the message $m_0$ in round $0$. Assume all nodes at distance at most $i$ from $s$ in $G \setminus F$ accepts the message $m_0$ by round $i$. Consider a node $v$ at distance $(i+1)$ from $s$ in $G \setminus F$. By the induction assumption on $i$, $v$ receives the message $accept(m_0)$ from a neighbor $w$ in round $j \leq (i+1)$ over a reliable edge $(w,v)$. 
	We are left to show that $\mincut(s,v, \mathcal{P})\geq t$, where $\mathcal{P}$ is as given by Eq. (\ref{eq:path-col}). Alternatively, we show that for every edge set $E' \subseteq E \setminus \{(w,v)\}$ of size $(t-1)$, the node $v$ stores a heard bundle containing $m_0$ and a path $P_k$ such that $P_k \cap (E'\cup \{(v,w)\}) = \emptyset$ during the first phase. This necessary implies that the minimum cut is at least $t$.
	
	For a subset $E' \subseteq E$ of size $(t-1)$, as $|F \cup E'\cup \{w,v\}| \leq 2t$ by the promise on $L$ in \Cref{lem:main-t}, $\dist_{G \setminus (F \cup E'\cup \{w,v\})}(s,v) \leq L$. By the properties of the covering family $\mathcal{G}$ (Def. \ref{def:covering-Lk}), it follows that there exists a subgraph $G_k$ such that $G_k \cap (F \cup E' \cup \{(v,w)\})=\emptyset$, and $\dist_{G_k}(s,v) \leq L$. Hence, all edges in $G_k$ are reliable, and the only message that passed through the heard bundles during the $k$'th iteration is the correct message $m_0$.
	As $\dist_{G_k}(s,v) \leq L$, the node $v$ stores a heard bundle $heard(m_0, x , P_k)$ during the $k$'th iteration for some $s$-$v$ path $P_k$ of length $x= O(L)$. 
	Moreover, as $P_k \subseteq G_k$ it also holds that $P_k \cap (E'\cup \{(v,w)\}) = \emptyset$.
	We conclude that $\mincut(s,v, \mathcal{P})\geq t$, and by the definition of Phase 2, $v$ accepts $m_0$ by round $(i+1)$. The claim follows as the diameter of $G \setminus F$ is $O(L)$ due to the promise on $L$.
\end{proof}

\paragraph{Algorithm $\UnBBalgo$ (Proof of \Cref{thm:t-broadcast})}
We now describe the general broadcast algorithm. 
Our goal is to apply Alg. $\BBtalgo(\cG,L,t)$ over the $(L,2t)$ covering family $\cG$ for $L=O(t D)$, constructed using \Cref{fc:rpc-t}. Since the nodes do not know the diameter $D$ (or a linear estimate of it), we make $O(\log D)$ applications of Alg. $\BBtalgo(\cG_i,L_i,t)$ using the $(L_i,2t)$ covering family $\cG_i$ for $L_i=O(t D_i)$, where $D_i=2^i$ is the diameter guess for the $i$'th application. 

Specifically, at the beginning of the $i$'th application, the source $s$ initiates the execution of Alg.
$\BBtalgo(\cG_i,L_i,t)$ with the desired message $m_0$ over an $(L_i,2t)$ covering family $\mathcal{G}_i$ constructed using \Cref{fc:rpc-t} with $L_i=O(t D_i)$ and $D_i=2^{i}$. 
Denote all nodes that accepted the message $m_0$ at the end of Alg. $\BBtalgo(\cG_i,L_i,t)$ by $A_i$, and let $N_i=V \setminus A_i$ be the nodes that did not accept the message. 

The algorithm now applies an additional step where the nodes in $N_i$ inform $s$ that they did not accept any message in the following manner. All nodes in $N_i$ broadcast \emph{the same} designated message $M$ by applying Alg. $\BBtalgo(\cG'_i,c t L_i,t)$ over a $(c t L_i,2t)$ covering family $\cG'_i$, for some fixed constant $c>0$ (known to all nodes). This can be viewed as performing a single broadcast execution (i.e., with the same source message) but from $|N_i|$ multiple sources. We next set $\tau_i=O(t \cdot D_i \log n)^{O(t)}$ as a bound on the waiting time for a node to receive any acknowledgment. 

If the source node $s$ accepts the message $M$ at the end of this broadcast execution, it waits for $\tau_i$ rounds, and then continues to the next application\footnote{We make the source node $s$ wait since in the case where it actually sends a termination message, all nodes accept it within $\tau_i$ rounds. Therefore, we need to make sure that all nodes start the next $(i+1)$ application at the same time.} $(i+1)$ (with diameter guess $2^{i+1}$). In the case where $s$ did not accept the message $M$ within $\tau_i$ rounds from the beginning of that broadcast execution, it broadcasts a termination message $M_T$ to all nodes in $V$. This is done by applying Alg. $\BBtalgo(\cG'_i,c t L_i,t)$ over the $(c t L_i,2t)$ covering family $\cG'_i$. Once a node $v \in V$ accepts the termination message $M_T$, it completes the execution with the last message it has accepted so far (in the analysis part, we show that it indeed accepts the right message). A node $v$ that did not receive a termination message $M_T$ within $\tau_i$ rounds, continues to the next application of Alg. $\BBtalgo$.

The correctness argument exploits the fact that for an application $i$ such that $N_i \neq \emptyset$, the graph $G'$ obtained by contracting \footnote{I.e., we contract all edges with both endpoints in $N_i$.} all nodes in $N_i$ into a single node $a$, satisfies the following: (i) it is $(2t+1)$ edge-connected, (ii) it contains $s$, and (iii) it has diameter $O(L_i)=O(t \cdot D_i)$. 
Property (i) holds since $G$ is $(2t+1)$ edge-connected, and property (ii) holds by the definition of $A_i$.
We next show property (iii) holds as well.
\begin{claim} \label{clm:distance-t}
	For every application $i$ of Alg. $\BBtalgo(\cG_i,L_i,t)$ and every node $v \in A_i$, it holds that $\dist_{G \setminus F}(s,v) = O(L_i)$.
\end{claim}
\begin{proof}
	Recall that Alg. $\BBtalgo$ proceeds in two phases. In the first phase, the source node propagates \emph{heard} bundles, and in the second phase, the source node propagates \emph{accept} messages.  
	Let $v$ be a node that accepts the message $m_0$ in the $i$'th application of Alg. $\BBtalgo(\cG_i,L_i,t)$. By Phase 2 of Alg. $\BBtalgo$, $v$ received an $accept(m_0)$ message from a neighbor $w$ in the second phase, and in the first phase it received heard bundles regarding $m_0$ over a path collection $\cP$ such that $\mincut(s,v,\cP) \geq t$, and $(v,w) \notin P$ for every $P \in \cP$.
	Let $P_1$ be the path on which the message $accept(m_0)$ propagated towards $v$ in the second phase of Alg. $\BBtalgo(\cG_i,L_i,t)$, executed by $s$ with the desired message $m_0$. Since the second phase is executed for $O(L_i)$ rounds, it holds that $|P_1| = O(L_i)$. If $P_1 \cap F = \emptyset$, $P_1$ starts at $s$ and the claim follows.

	Next, assume $P_1 \cap F \neq \emptyset$. Let $(v_1,v_2)$ be the closest edge to $v$ on $P_1$ such that $(v_1,v_2)\in F$ (i.e., $P_1[v_2,v] \cap F = \emptyset$). Hence, $v_2$ sent an $accept(m_0)$ message to $v$ during the second phase over the path $P_1$. Therefore, $v_2$ received heard bundles regarding $m_0$ over a path collection $\cP$ such that $\mincut(s,v,\cP) \geq t$, and $(v_1,v_2) \notin P$ for every $P \in \cP$. As $|F \setminus \{(v_1,v_2)\}| \leq (t-1)$, there must be a path $P_2 \in \cP$ such that $P_2 \cap F = \emptyset$. Hence, $P_2$ is an $s$-$v_2$ path  of length $|P_2| = O(L_i)$ in $G \setminus F$. Thus, the concatenated path $P_2 \circ P_1[v_2,v]$ is a path of length $O(L_i)$ from $s$ to $v$ in $G \setminus F$. 
\end{proof}

Using similar arguments to \Cref{obs:bounded-diam-f}, from \Cref{clm:distance-t} it follows that for every set $E' \subseteq E$ of size $|E'| \leq t$, and every node $v \in A_i$ it holds that $\dist_{G \setminus F\cup E'}(N_i,v) = O(t \cdot L_i)$.
\begin{observation} \label{obs:distance-2t+1}
	For every node $v \in A_i$ and an edge set $\E' \subseteq E$ of size $|E'|\leq t$, it holds that $\dist_{G \setminus (F\cup E')}(N_i,v) = O(t \cdot L_i)$.
\end{observation} 
\begin{proof}
Let $T$ be a truncated BFS tree rooted at $s$ in $G \setminus F$, such that (1) $A_i \subseteq V(T)$, and (2) all the leaf nodes of $T$ are included in $A_i$. 
By \Cref{clm:distance-t} the depth of $T$ is $O(L_i)$.
It then follows that the forest $T \setminus  E'$ consists of $|E'|+1\leq (t+1)$ trees, each of diameter $O(L_i)$. 
Since $G$ is $(2t+1)$ edge-connected, there exists a path from some node in $N_i$ to $v$ in $G \setminus (F\cup E')$.
Hence, the shortest path from $N_i$ to $v$ in $G \setminus (F\cup E')$ denoted as $P$ can be transformed into a path $P'$ containing at most $t$ edges in $P$, as well as, $(t+1)$ tree subpaths of the forest $T \setminus E'$. Therefore, $\dist_{G \setminus (F\cup E')}(N_i,v) =O(t \cdot L_i)$.
\end{proof}

We are now ready to complete the proof of \Cref{thm:t-broadcast}.

\begin{proof}[Proof of \Cref{thm:t-broadcast}]
Fix an application $i$ of Alg. $\BBtalgo(\cG_i,L_i,t)$. 
We first claim that if $N_i \neq \emptyset$, the source node $s$ accepts the message $M$, and continues to the next application $(i+1)$.
In the second step of the $i$'th application, the $N_i$ nodes broadcast the message $M$ by executing Alg. $\BBtalgo(\cG'_i,ctL_i,t)$ over a $(ctL_i,2t)$ covering family $\cG'_i$. 
By \Cref{obs:distance-2t+1}, for a large enough constant $c$ it hold that for every $v \in A_i$ and $\E' \subseteq E$ of size $|E'|\leq t$, it holds that $\dist_{G \setminus F\cup E'}(N_i,v) \leq c t L_i$. 
Hence, using the same arguments as in \Cref{clm:t-all-accepts} we can conclude that all nodes in $A_i$ accept the message $M$ in the second phase of Alg. $\BBtalgo(\cG'_i,ctL_i,t)$.
Specifically, the source node $s\in A_i$ accepts the message and proceeds to application $(i+1)$. 

Next, consider the first application $i^*$ for which $N_{i^*} = \emptyset$, and all nodes accept the broadcast message $m_0$. Since the diameter of the graph is $D$, the algorithm must have run at least $D/2$ rounds. This implies that $D_{i^*}=\Omega(D/t)$ and $L_{i^*} = \Omega(D)$.  Hence, for a large enough constant $c$, when $s$ broadcasts the termination message $M_T$ using Alg. $\BBtalgo(\cG'_{i^*},ctL_{i^*},t)$ over the $(ctL_{i^*},2t)$ covering family $\cG'_{i^*}$, all nodes accept $M_T$ and finish the execution. 
	
	We next consider the round complexity. Clearly, when taking $D_i \geq D$, by \Cref{lem:main-t} the set $N_i=\emptyset$, and we are done. Therefore, we can conclude that $i^* \leq \log D + 1$. 
	The algorithm performs $O(\log D)$ applications of Alg. $\BBtalgo$, each with parameters $L_i=O(tD)$ and $ctL_i=O(t^2D)$ (used in the construction of $\cG'_i$). Thus, the total round complexity is bounded by $(tD \log n)^{O(t)}$.
\end{proof}
Finally, we observe that our broadcast algorithm can be implemented in the \local\ model using $O(tD\log n)$ many rounds.
\begin{corollary}\label{cor:localByzn}
	For every $(2t+1)$ edge-connected graph, and a source node $s$, there is a deterministic broadcast algorithm against $t$ adversarial edges that runs in $O(tD \log n)$ \emph{local} rounds.
\end{corollary}
\begin{proof}
	The algorithm is the same as in the \congest\ model. However, since in the local model there are no bandwidth restrictions, the message propagation over the $|\mathcal{G}|$ subgraphs of the $(L, t)$ covering family can be implemented simultaneously within $L=O(tD)$ rounds. 
\end{proof}

\section{Broadcast Algorithms for Expander Graphs}\label{sec:expanders}

In this section, we show improved constructions of covering families for expander graphs, which consequently lead to considerably faster broadcast algorithms. 
Specifically, we show that the family of expander graphs with a sufficiently large minimum degree admits covering families of a considerably smaller size (than that obtained for general graphs). 
We start by providing the precise definition of expander graphs used in this section. 

\paragraph{Expander graphs.} For a node subset $S \subseteq V$, denote by $\delta_G(S)=|E \cap (S \times (V\setminus S))|$ the number of edges crossing the $(S, V \setminus S)$ cut . The \emph{volume} $vol_G(S)$ of $S$ in $G$ is the number of edges of $G$ incident to $S$. Assuming both $S$ and $V \setminus S$ have positive volume in $G$, the \emph{conductance} $\phi_G(S)$ of $S$ is defined as $\phi_G(S)=\delta_G(S)/\min\{ vol_G(S), vol_G(V \setminus S)\}$. The \emph{edge expansion}  of a graph $G$ is given by $\phi(G)=\min_{S\subset V}\phi_G(S)$. We say a graph $G$ is a $\phi$-expander if $\phi(G)\geq \phi$.

The structure of this section is as follows. We first provide a combinatorial randomized construction of $(L, 2t)$ covering families for $n$-node $\phi$-expander graphs\footnote{with slightly weaker properties, which are sufficient for our purposes.} with minimum degree $\Theta(t \log n/ \phi)$. Then, we show a randomized construction of $(L, 2t)$ covering families in the adversarial \congest\ model. Finally, we provide an improved broadcast algorithm that uses these families and is resilient against $t$ adversarial edges, given a $\phi$-expander graph with minimum degree $\Theta(t^2 \log n / \phi)$.

\paragraph{Covering families with improved bounds.}  The computation of the improved covering family is based on showing that a sampled subgraph of $G$ obtained by sampling each edge independently with probability $p=\Theta(1 / t)$ satisfies some desired expansion and connectivity properties. 

We use the following result from \cite{wulff2017fully} which provides conductance guarantees for the subgraphs obtained by Karger's edge sampling technique. This result also implies that expander graphs with large minimum degree have large edge-connectivity.

\begin{theorem}[Lemma 20 from \cite{wulff2017fully}]\label{thm:conduct-karger}
	Given $c>0, \kappa \geq 1$ and $\rho\leq 1$, let $G=(V,E)$ be an $n$-node multigraph with degree at least $\kappa \rho$. Let $G'=(V,E')$ be the multigraph obtained from $G$ by sampling each edge independently with probability 
	$$p=\min\{1, (12c+24)\ln n/(\rho^2 \kappa)\}~.$$
	Then, with probability $1-O(1/n^c)$, for every cut $(S,V \setminus S)$ in $G$, it holds that:
	\begin{itemize}
		\item if $\phi_{G}(S)\geq \rho$ then $\phi_{G'}(S)$ deviates from $\phi_{G}(S)$ by a factor of at most $4$, and
		\item if $\phi_{G}(S)< \rho$ then $\phi_{G'}(S)<6\rho$.
	\end{itemize}
\end{theorem}

\begin{corollary}
	Given a $\rho$-expander graph $G$ with minimum degree $\kappa \cdot \rho$, the edge-connectivity of $G$ is $\Omega(\rho^2 \kappa/ \log^2 n)$.
\end{corollary}
\begin{proof}
For $p=c_1\log n/ \rho^2 \kappa$, consider a collection of $\lceil 1/p \rceil$ subgraphs $\mathcal{H}= \{H_1, \ldots H_{\lceil 1/p \rceil}\}$, where every subgraph $H_i$ is obtained by sampling each $G$-edge into $H_i$ independently with probability $p$. Our first goal is to show that with high probability the following properties hold: (i) all subgraphs in $\mathcal{H}$ are connected, and (ii) every edge $e\in E$ appears in at most $c_2 \log n$ many subgraphs in $\mathcal{H}$ for some constant $c_2 > 0$. 

 By \Cref{thm:conduct-karger} for $c_1 \geq 12c+24$ each subgraph $H_i\in \mathcal{H}$ has conductance $\Theta(\rho)$ with probability $1-1/n^c$. 
 Hence, for a large enough constant $c_1$, by the union bound over the $\lceil 1/p \rceil= \lceil \rho^2 \kappa/c_1 \log n \rceil$ subgraphs in $\mathcal{H}$, all subgraphs have conductance $\Theta(\rho)$ w.h.p. Hence, w.h.p. every subgraph $H_i\in \mathcal{H}$ is connected and property (i) holds.
Consider property (ii). Since we have $\lceil 1/p \rceil$ many subgraphs in $\mathcal{H}$, by the Chernoff bound w.h.p. every edge $e\in E$ appears in at most $O(\log n)$ subgraphs in $\mathcal{H}$. 
Thus, we conclude that w.h.p. $\mathcal{H}$ satisfies both properties (i) and (ii).

Since a random collection of $\lceil 1/p \rceil$ subgraphs as descried above satisfies the properties w.h.p, we conclude that $G$ must contain a collection $\mathcal{H}^*$ of $\lceil 1/p \rceil$ subgraphs that satisfies the properties.
We now complete the argument on the edge-connectivity of $G$ by showing that for every subset $E' \subseteq E$ of size at most $|E'|=\ell$ for $\ell=\lfloor 1/(2c_2\log n\cdot p) \rfloor =\Theta(\rho^2 \kappa/ \log^2 n)$, the subgraph $G \setminus E'$ is connected. By property (ii) of $\mathcal{H}$, every edge $e\in E'$ appears in at most $c_2 \log n$ subgraphs. Hence, the total number of subgraphs in $\mathcal{H}$ that includes edges from $E'$ is at most $c_2 \log n \cdot \ell \leq \lceil 1/p \rceil /2$. It follows that there exists a subgraph $H_i\in \mathcal{H}$ that avoids all edges in $E'$ (i.e., $H_i \subseteq G \setminus E'$). By property (i) the subgraph $H_i$ is connected, and therefore $G \setminus E'$ is connected as well. The claim follows.
\end{proof}

We will use the following fact that appeared in \cite{spielman09} that provides an upper bound for the diameter of $\phi$-expander graphs. For completeness we provide here the proof.
\begin{fact}[Section 19.1 from \cite{spielman09}]\label{thm:bounded-diam-conduct}
	The diameter of an $n$-node $\phi$-expander graph $G$ is bounded by $O(\log n / \phi)$.
\end{fact}
\begin{proof}
	For an integer $k \geq 0$ and a node $w$, let $B_k(w)=\{v\in V~\mid~ \dist_G(w,v)\leq k \}$ be the set of nodes at distance at most $k$ from $w$. We show by induction on $k\geq 0$ that for every node $w$ satisfying that $vol_G(B_k(w)) \leq |E|/2$, it holds that $vol_G(B_k(u)) \geq (1+\phi)^k$. For the base case of $k=0$, $B_0(w)=\{w\}$. Since $\phi>0$, the node $w$ has at least one neighbor in $G$, and thus $vol_G(B_0(w)) \geq 1$. 
	Assume that the claim holds for $k$ and consider $B_{k+1}(w)$ such that $vol_G(B_{k+1}(w)) \leq |E|/2$. 
	By the assumption, it holds that $vol_G(B_{k+1}(w))) \leq vol_G(V \setminus B_{k+1}(w))$, and therefore:
	\begin{equation} \label{eq:1}
		\delta_G(B_{k+1}(w))\geq \phi \cdot vol_G(B_{k+1}(w))~.
	\end{equation}
	By definition of $vol_G(S)$, we have:
	\begin{equation} \label{eq:2}
		vol_G(B_{k+1}(w)) \geq vol_G(B_{k}(w))+\delta_G(B_{k+1}(w))~.
	\end{equation}
	Combining \Cref{eq:1} and \Cref{eq:2}, we conclude that 
	\begin{align}
		vol_G(B_{k+1}(w)) &\geq vol_G(B_{k}(w))+\delta_G(B_{k+1}(w))  \nonumber \\ 
		&\geq vol_G(B_{k}(w)) +  \phi \cdot vol_G(B_{k+1}(w)) \nonumber \\ 
		&\geq (1+\phi) \cdot vol_G(B_{k}(w)) \nonumber \\
		&\geq (1+\phi)^{k+1}, \nonumber 
	\end{align}
	where the last inequality follows by the induction assumption.
	Let $k$ be the minimal integer satisfying that $(1+\phi)^{k} \geq |E|/2$, and consider a fixed pair of nodes $u$ and $v$. By the above, we have that $|vol_G(B_k(u))|, |vol_G(B_k(v))|\geq |E|/2$. Thus, there is an edge $(w_1,w_2)$ where $w_1 \in B_k(u)$ and $w_2 \in B_k(v)$, and therefore $\dist_G(u,v)\leq 2k+1=O(\log n/\phi)$. 
\end{proof}
By combining \Cref{thm:conduct-karger} with \Cref{thm:bounded-diam-conduct} we get:
\begin{corollary}\label{cor:expander-connec-sampling}
	Let $G=(V,E)$ be an $n$-node $\phi$-expander graph with minimum degree $\Theta(t \log n/\phi)$. 
	Let $G'$ be a subgraph of $G$ obtained by sampling each $G$-edge independently into $G'$ with probability of $p=\Theta(1/t)$. Then $G'$ has diameter of $O(\log n /\phi)$, w.h.p.
\end{corollary}
\begin{proof}
	Since the minimum degree of $G$ is $\Delta = \Theta(t \log n/\phi)$, it holds that $\Delta \geq \kappa \cdot \phi$ for $\kappa = t \log n / \phi^2$, and  $p=\Theta(\log n /(\kappa \cdot \phi^2))$.
	Therefore, by \Cref{thm:conduct-karger} it holds that $\phi(G')=\Theta(\phi)$, w.h.p. By~\Cref{thm:bounded-diam-conduct}, we have that the diameter of $G'$ is $O(\log n /\phi)$.
\end{proof}
We next show that for $\phi$-\emph{expander} graphs with minimum degree $\Theta(t \log n/\phi)$,
there exist covering graph families with a considerable smaller cardinality than those obtained for \emph{general} graphs. Note that the definition of the covering family in \Cref{thm:cover-family-expander} is more relaxed than that of \Cref{def:covering-Lk} w.r.t property (P1). 
\begin{lemma}\label{thm:cover-family-expander}
	Let $G$ be an $n$-node $\phi$-expander graph with minimum degree $\Theta(t\log n/\phi)$.
	There is a randomized algorithm that computes an $(L,2t)$ covering family $\mathcal{G}$ of $O(t\cdot \log n)$ subgraphs for $L=O(\log n /\phi)$, satisfying the following properties with high probability. 
	For every $u,v,E' \in V \times V \times E^{\leq 2t}$, there exists a subgraph $G_i$ such that (P1') $\dist_{G_i}(u, v) \leq L$ and (P2) $E' \cap G_i=\emptyset$.
\end{lemma}
\begin{proof}
The covering family $\mathcal{G}$ is obtained by sampling a collection of $O(t\cdot \log n)$ subgraphs. 
	Specifically, each subgraph $G_i \in \mathcal{G}$ is obtained by sampling each $G$-edge into $G_i$ independently with probability of $p=\Theta(1/t)$. By \Cref{cor:expander-connec-sampling}, it holds that w.h.p. $G_i$ is a connected graph with diameter $O(\log n /\phi)$. We now show that w.h.p. $\mathcal{G}$ is an $(L,2t)$ covering family for $L=O(\log n /\phi)$.
	
	Fix a pair of nodes $u,v\in V$ and a subset of at most $2t$ edges $E' \subseteq G$. We claim that with probability of at least $1-1/n^{\Omega(t)}$, there exists a subgraph $G_i$ satisfying (P2). 
	The probability that $E' \cap G_i=\emptyset$ is at least $q= (1-p)^{2t}\geq 1/e^3$. Therefore, the probability that no subgraph in $\mathcal{G}$ satisfies (P2) is at most $(1-q)^{c\cdot t \log n}\leq 1/n^{c't}$.
	By applying the union bound over all $O((tn)^{4t})$ possible subsets of $2t$ edges, we get that w.h.p. property (P2) holds for every subset $E'$. In addition, by \Cref{cor:expander-connec-sampling} w.h.p. it holds that every $G_i$ has diameter $L=O(\log n /\phi)$. Therefore, w.h.p., for every pair of nodes $u,v\in V$ it holds that $\dist_{G_i}(u, v) \leq L$. By the union bound over the subgraphs in $\mathcal{G}$, we get that w.h.p. both properties (P1') and (P2) hold for every $u,v,E' \in V \times V \times E^{\leq 2t}$.
\end{proof}

Unfortunately, it is unclear how to compute the covering family of \Cref{thm:cover-family-expander} in the adversarial \congest\ model. The reason is that the endpoints of an adversarial edge $e=(u,v) \in F$ cannot faithfully sample $e$ with probability $p$. For example, letting the endpoint of larger ID $u$ perform the sampling, the other endpoint $v$ might not be correctly informed with the outcome of this sampling. To resolve this, we let each endpoint sample a directed edge $(u,v)$ with a probability of $p$. Consequently, the graph family that we get consists of directed graphs, and the graph is required to have minimum degree $\Theta(t^2 \log n / \phi)$. 

\begin{lemma}\label{theorem:cover-family-expander-directed}[Directed Covering Families for Expanders]
	For any $n$-node $\phi$-expander graph $G=(V,E)$ with minimum degree $\Theta(t^2\log n/\phi)$, there is a randomized distributed algorithm that in the presence of at most $t$ adversarial edges $F$, locally computes a \emph{directed} $(L,2t)$ covering family of $O(t \log n)$ \emph{directed} subgraphs $\mathcal{G}$ for $L=O(\log n /\phi)$, satisfying the following w.h.p.
	For every $u,v,E' \in V \times V \times E^{\leq 2t}$, there exists a subgraph $G_i$ such that (P1') $\dist_{G_i}(u, v) \leq L$ and (P2) $E' \cap G_i=\emptyset$ (i.e., $G_i$ does not contain any edge in $E'$ in neither direction).
\end{lemma}
\begin{proof}
	The covering family $\mathcal{G}$ is obtained by sampling a collection of $O(t \log n)$ \emph{directed} subgraphs. 
	Specifically, each subgraph $G_i \in \mathcal{G}$ is obtained by sampling a directed edge $(u,v)$ with probability of $p=\Theta(1/t)$. The sampling of a directed edge $(u,v)$ is done locally by $v$.
	We now show that w.h.p. $\mathcal{G}$ is a (directed) $(L,2t)$ covering family for $L=O(\log n /\phi)$. 
	
	Fix a pair of nodes $u,v\in V$ and a subset of at most $2t$ edges $E' \subseteq G$. 
	We claim that with a probability of at least $1-1/n^{\Omega(t)}$, there exists a subgraph $G_i$ satisfying (P2). 
	The probability that $E' \cap G_i=\emptyset$ is at least $q=(1-p)^{4t}\geq 1/e^5$. Therefore, the probability that no subgraph in $\mathcal{G}$ satisfies (P2) is at most
	$(1-q)^{c\cdot t \log n}=1/n^{c't}$. By applying the union bound over all $O((tn)^{4t})$ possible sets of $2t$ edges, we get that w.h.p. property (P2) holds for every subset $E'$. 
	
	As each directed edge is sampled with probability $p$, each edge is sampled in both directions with probability $p^2=\Theta(1/t^2)$. Since $G$ has minimum degree $\Theta(t^2 \log n/ \phi)$, by 
	 \Cref{cor:expander-connec-sampling} it holds that w.h.p. $G_i$ contains a bidirected subgraph which (when viewed as an undirected subgraph) has diameter $O(\log n /\phi)$. 
	As a result,  w.h.p. it holds that every $G_i$ has a round-trip diameter $O(\log n /\phi)$. Therefore, for $L=\Theta(\log n / \phi)$ w.h.p. for every pair of nodes $u,v\in V$ it holds that $\dist_{G_i}(u, v) \leq L$. By the union bound over the subgraphs in $\mathcal{G}$, we get that w.h.p. both properties (P1') and (P2) hold for every $u,v,E' \in V \times V \times E^{\leq 2t}$.
\end{proof}
\paragraph{Broadcast algorithm using a directed covering family (Proof of \Cref{thm:expander-broadcast}).} Let $G$ be an $n$-node $\phi$-expander graph, with minimum degree $\Theta(t^2\log n/\phi)$, and a fix set of $t$ unknown adversarial edges $F$. 
Throughout the algorithm, we assume that the nodes hold a linear estimate om the expansion parameter $\phi$. (This assumption can also be avoided by adding a logarithmic overhead in the graph diameter to the final round complexity.)

\noindent The algorithm begins with locally computing an $(L,2t)$ directed covering family $\mathcal{G}=\{G_1, \ldots, G_{\ell}\}$ for $G$ of size $\ell=O(t \log n)$ using \Cref{theorem:cover-family-expander-directed}, where $L=O(\log n /\phi)$. 
Next, the nodes execute the broadcast algorithm $\BBtalgo(L,t)$ of \Cref{lem:main-t} 
over $\mathcal{G}$.

Recall that the algorithm of \Cref{lem:main-t} consists of two phases, a flooding phase where heard bundles are propagated over the subgraphs in the covering family, and an acceptance phase.
The flooding phase proceeds in $\ell$ iterations, where each iteration is implemented in $O(L)$ rounds.
In every iteration $i$, the nodes propagate heard bundles over the directed subgraph $G_i \in \mathcal{G}$. Specifically, a node $v\in V$ stores and sends only heard bundles that are received over its \emph{directed} incoming edges, which are sampled (locally by $v$) into $G_i$.
Every message received by $v$ from a neighbor $u\in N(v)$, such that $(u,v)\notin G_i$ is ignored. 
The acceptance phase is executed in $O(L)$ rounds, as described in~\Cref{sec:t-faults}.

\paragraph{Correctness.}
We now show that the correctness of the algorithm still holds. Let $s$ be the designated source node, and let $m_0$ be the message $s$ sends using the broadcast algorithm. We first note that due to \Cref{clm:t-no-false}, no node $v \in V$ accepts a wrong message $m' \neq m_0$.
We are left to show all nodes in $V$ accept the message $m_0$ during the second phase. 
\begin{claim}
	All nodes accept $m_0$ within $O(L)$ rounds from the beginning of the second phase.
\end{claim}
\begin{proof} \vspace{-5pt}
We show that all nodes accept the message $m_0$ by induction on the distance from the source $s$ in the graph $G \setminus F$.	
The base case holds since $s$ accepts the message $m_0$ in round $0$. Assume all nodes at distance at most $i$ from $s$ in $G \setminus F$ accepts the message by round $i$. Consider a node $v$ at distance $(i+1)$ from $s$ in $G \setminus F$. By the induction assumption on $i$, the node $v$ receives the message $accept(m_0)$ from a neighbor $w$, in round $j \leq (i+1)$ over a  reliable edge $(w,v)$. 
We are left to show $\mincut(s,v, \mathcal{P})\geq t$, where $\mathcal{P}$ is as given by Eq. (\ref{eq:path-col}): 
	\begin{equation*}
		\mathcal{P}=\{P ~\mid~ \text{$v$ stored a $heard(m',len,P)$  message and } (v,w),(w,v) \notin P\}~.
	\end{equation*}
	Alternatively, we show that for every edge set $E' \subseteq E \setminus \{(w,v)\}$ of size $(t-1)$, the node $v$ stores a heard bundle containing $m_0$ and a path $P_k$ such that $P_k \cap (E' \cup \{(w,v)\}) = \emptyset$ during the first phase. This necessary implies that the minimum-cut is at least $t$.
	
	By \Cref{theorem:cover-family-expander-directed} there exists a subgraph $G_k\in \mathcal{G}$ satisfying: (P1') $dist_{G_k}(s,v) \leq L$, and (P2) $G_k \cap (F \cup E' \cup \{(v,w)\})=\emptyset$. Hence, all directed edges in $G_k$ are reliable, and the only message passed through the heard bundles during the $k$'th iteration is the correct message $m_0$. Additionally, as  $G_k$ contains a directed $s$-$v$ path of length $O(L)$, the node $v$ stores a heard bundle $heard(m_0, x , P_k)$ during the $k$'th iteration, for some $s$-$v$ path $P_k$ of length $x$. As $P_k \subseteq G_k$, it also holds that $P_k \cap (E'\cup \{(v,w)\}) = \emptyset$.
	We conclude that $\mincut(s,v, \mathcal{P})\geq t$, and by the definition of Phase 2, $v$ accepts $m_0$ by round $(i+1)$.  Since $|F| \leq t$, \Cref{theorem:cover-family-expander-directed} implies that the diameter of $G \setminus F$ is $O(L)$, and the claim follows.
\end{proof}

\noindent\textbf{Round complexity.} The first phase consists of $\ell=O(t \log n)$ iterations, each implemented using $O(L)=O(\log n /\phi)$ rounds. The second phase takes $O(t \log n / \phi)$ rounds. Hence the total round complexity is bounded by $O(t \cdot \log^2 n/\phi)$.

\paragraph{Acknowledgments.} We are very grateful to David Peleg and Eylon Yogev for many useful discussions. 
\bibliographystyle{plain}
\bibliography{secure-addition}

\begin{thebibliography}{10}

\bibitem{AbrahamCDNP0S19}
Ittai Abraham, T.{-}H.~Hubert Chan, Danny Dolev, Kartik Nayak, Rafael Pass,
  Ling Ren, and Elaine Shi.
\newblock Communication complexity of byzantine agreement, revisited.
\newblock In {\em Proceedings of the 2019 {ACM} Symposium on Principles of
  Distributed Computing, {PODC} 2019, Toronto, ON, Canada, July 29 - August 2,
  2019}, pages 317--326, 2019.

\bibitem{AbrahamDDN019}
Ittai Abraham, Srinivas Devadas, Danny Dolev, Kartik Nayak, and Ling Ren.
\newblock Synchronous byzantine agreement with expected {O(1)} rounds, expected
  o(n\({}^{\mbox{2}}\)) communication, and optimal resilience.
\newblock In {\em Financial Cryptography and Data Security - 23rd International
  Conference, {FC} 2019, Frigate Bay, St. Kitts and Nevis, February 18-22,
  2019, Revised Selected Papers}, pages 320--334, 2019.

\bibitem{alon1995color}
Noga Alon, Raphael Yuster, and Uri Zwick.
\newblock Color-coding.
\newblock {\em Journal of the ACM (JACM)}, 42(4):844--856, 1995.

\bibitem{AugustineP015}
John Augustine, Gopal Pandurangan, and Peter Robinson.
\newblock Fast byzantine leader election in dynamic networks.
\newblock In {\em Distributed Computing - 29th International Symposium, {DISC}
  2015, Tokyo, Japan, October 7-9, 2015, Proceedings}, pages 276--291, 2015.

\bibitem{bagchi1994information}
Anindo Bagchi and S.~Louis Hakimi.
\newblock Information dissemination in distributed systems with faulty units.
\newblock {\em IEEE Transactions on Computers}, 43(6):698--710, 1994.

\bibitem{BermanG93}
Piotr Berman and Juan~A. Garay.
\newblock Cloture votes: n/4-resilient distributed consensus in t+1 rounds.
\newblock {\em Math. Syst. Theory}, 26(1):3--19, 1993.

\bibitem{BermanGP89}
Piotr Berman, Juan~A. Garay, and Kenneth~J. Perry.
\newblock Towards optimal distributed consensus (extended abstract).
\newblock In {\em 30th Annual Symposium on Foundations of Computer Science,
  Research Triangle Park, North Carolina, USA, 30 October - 1 November 1989},
  pages 410--415. {IEEE} Computer Society, 1989.

\bibitem{blough1993optimal}
Douglas~M Blough and Andrzej Pelc.
\newblock Optimal communication in networks with randomly distributed byzantine
  faults.
\newblock {\em Networks}, 23(8):691--701, 1993.

\bibitem{bodwin2021optimal}
Greg Bodwin, Michael Dinitz, and Caleb Robelle.
\newblock Optimal vertex fault-tolerant spanners in polynomial time.
\newblock In {\em Proceedings of the 2021 ACM-SIAM Symposium on Discrete
  Algorithms (SODA)}, pages 2924--2938. SIAM, 2021.

\bibitem{bracha87}
Gabriel Bracha.
\newblock Asynchronous byzantine agreement protocols.
\newblock {\em Information and Computation}, 75(2):130--143, 1987.

\bibitem{BrachaT85}
Gabriel Bracha and Sam Toueg.
\newblock Asynchronous consensus and broadcast protocols.
\newblock {\em J. {ACM}}, 32(4):824--840, 1985.

\bibitem{censor2017fast}
Keren Censor-Hillel and Tariq Toukan.
\newblock On fast and robust information spreading in the vertex-congest model.
\newblock {\em Theoretical Computer Science}, 2017.

\bibitem{ChakrabortyC20}
Diptarka Chakraborty and Keerti Choudhary.
\newblock New extremal bounds for reachability and strong-connectivity
  preservers under failures.
\newblock In {\em 47th International Colloquium on Automata, Languages, and
  Programming, {ICALP} 2020, July 8-11, 2020, Saarbr{\"{u}}cken, Germany
  (Virtual Conference)}, pages 25:1--25:20, 2020.

\bibitem{ChlebusKO20}
Bogdan~S. Chlebus, Dariusz~R. Kowalski, and Jan Olkowski.
\newblock Fast agreement in networks with byzantine nodes.
\newblock In {\em 34th International Symposium on Distributed Computing, {DISC}
  2020, October 12-16, 2020, Virtual Conference}, pages 30:1--30:18, 2020.

\bibitem{CoanW92}
Brian~A. Coan and Jennifer~L. Welch.
\newblock Modular construction of a byzantine agreement protocol with optimal
  message bit complexity.
\newblock {\em Inf. Comput.}, 97(1):61--85, 1992.

\bibitem{CohenHMOS19}
Ran Cohen, Iftach Haitner, Nikolaos Makriyannis, Matan Orland, and Alex
  Samorodnitsky.
\newblock On the round complexity of randomized byzantine agreement.
\newblock In {\em 33rd International Symposium on Distributed Computing, {DISC}
  2019, October 14-18, 2019, Budapest, Hungary}, pages 12:1--12:17, 2019.

\bibitem{dinitz2011fault}
Michael Dinitz and Robert Krauthgamer.
\newblock Fault-tolerant spanners: better and simpler.
\newblock In {\em Proceedings of the 30th annual ACM SIGACT-SIGOPS symposium on
  Principles of distributed computing}, pages 169--178. ACM, 2011.

\bibitem{Dolev82}
Danny Dolev.
\newblock The byzantine generals strike again.
\newblock {\em J. Algorithms}, 3(1):14--30, 1982.

\bibitem{DolevFFLS82}
Danny Dolev, Michael~J. Fischer, Robert~J. Fowler, Nancy~A. Lynch, and
  H.~Raymond Strong.
\newblock An efficient algorithm for byzantine agreement without
  authentication.
\newblock {\em Information and Control}, 52(3):257--274, 1982.

\bibitem{DolevH08}
Danny Dolev and Ezra~N. Hoch.
\newblock Constant-space localized byzantine consensus.
\newblock In {\em Distributed Computing, 22nd International Symposium, {DISC}
  2008, Arcachon, France, September 22-24, 2008. Proceedings}, pages 167--181,
  2008.

\bibitem{DworkPPU88}
Cynthia Dwork, David Peleg, Nicholas Pippenger, and Eli Upfal.
\newblock Fault tolerance in networks of bounded degree.
\newblock {\em {SIAM} J. Comput.}, 17(5):975--988, 1988.

\bibitem{FeldmanM97}
Pesech Feldman and Silvio Micali.
\newblock An optimal probabilistic protocol for synchronous byzantine
  agreement.
\newblock {\em {SIAM} J. Comput.}, 26(4):873--933, 1997.

\bibitem{fischer1983consensus}
Michael~J Fischer.
\newblock The consensus problem in unreliable distributed systems (a brief
  survey).
\newblock In {\em International Conference on Fundamentals of Computation
  Theory}, pages 127--140. Springer, 1983.

\bibitem{fitzi2000partial}
Mattias Fitzi and Ueli Maurer.
\newblock From partial consistency to global broadcast.
\newblock In {\em Proceedings of the thirty-second annual ACM symposium on
  Theory of computing}, pages 494--503, 2000.

\bibitem{ganesh2016broadcast}
Chaya Ganesh and Arpita Patra.
\newblock Broadcast extensions with optimal communication and round complexity.
\newblock In {\em Proceedings of the 2016 ACM Symposium on Principles of
  Distributed Computing}, pages 371--380, 2016.

\bibitem{GarayM98}
Juan~A. Garay and Yoram Moses.
\newblock Fully polynomial byzantine agreement for \emph{n} {\textgreater} 3t
  processors in \emph{t} + 1 rounds.
\newblock {\em {SIAM} J. Comput.}, 27(1):247--290, 1998.

\bibitem{Ghaffari15}
Mohsen Ghaffari.
\newblock Near-optimal scheduling of distributed algorithms.
\newblock In {\em Proceedings of the 2015 {ACM} Symposium on Principles of
  Distributed Computing, {PODC} 2015, Donostia-San Sebasti{\'{a}}n, Spain, July
  21 - 23, 2015}, pages 3--12, 2015.

\bibitem{grandoni2019faster}
Fabrizio Grandoni and Virginia~Vassilevska Williams.
\newblock Faster replacement paths and distance sensitivity oracles.
\newblock {\em ACM Transactions on Algorithms (TALG)}, 16(1):1--25, 2019.

\bibitem{0001W20}
Fabrizio Grandoni and Virginia~Vassilevska Williams.
\newblock Faster replacement paths and distance sensitivity oracles.
\newblock {\em {ACM} Trans. Algorithms}, 16(1):15:1--15:25, 2020.

\bibitem{ByzCompilersHP21}
Yael Hitron and Merav Parter.
\newblock General {CONGEST} compilers against adversarial edges.
\newblock In Seth Gilbert, editor, {\em 35th International Symposium on
  Distributed Computing, {DISC} 2021, October 4-8, 2021, Freiburg, Germany
  (Virtual Conference)}, volume 209 of {\em LIPIcs}, pages 24:1--24:18. Schloss
  Dagstuhl - Leibniz-Zentrum f{\"{u}}r Informatik, 2021.

\bibitem{imbs2015simple}
Damien Imbs and Michel Raynal.
\newblock Simple and efficient reliable broadcast in the presence of byzantine
  processes.
\newblock {\em arXiv preprint arXiv:1510.06882}, 2015.

\bibitem{karger1999random}
David~R Karger.
\newblock Random sampling in cut, flow, and network design problems.
\newblock {\em Mathematics of Operations Research}, 24(2):383--413, 1999.

\bibitem{RPC21}
{Karthik {C. S.}} and Merav Parter.
\newblock Deterministic replacement path covering.
\newblock In {\em Proceedings of the 2021 {ACM-SIAM} Symposium on Discrete
  Algorithms, {SODA} 2021, Virtual Conference, January 10 - 13, 2021}, pages
  704--723, 2021.

\bibitem{katz2006expected}
Jonathan Katz and Chiu-Yuen Koo.
\newblock On expected constant-round protocols for byzantine agreement.
\newblock In {\em Annual International Cryptology Conference}, pages 445--462.
  Springer, 2006.

\bibitem{KhanNV19}
Muhammad~Samir Khan, Syed~Shalan Naqvi, and Nitin~H. Vaidya.
\newblock Exact byzantine consensus on undirected graphs under local broadcast
  model.
\newblock In {\em Proceedings of the 2019 {ACM} Symposium on Principles of
  Distributed Computing, {PODC} 2019, Toronto, ON, Canada, July 29 - August 2,
  2019}, pages 327--336, 2019.

\bibitem{KingSSV06}
Valerie King, Jared Saia, Vishal Sanwalani, and Erik Vee.
\newblock Towards secure and scalable computation in peer-to-peer networks.
\newblock In {\em 47th Annual {IEEE} Symposium on Foundations of Computer
  Science {(FOCS} 2006), 21-24 October 2006, Berkeley, California, USA,
  Proceedings}, pages 87--98, 2006.

\bibitem{Koo04}
Chiu{-}Yuen Koo.
\newblock Broadcast in radio networks tolerating byzantine adversarial
  behavior.
\newblock In Soma Chaudhuri and Shay Kutten, editors, {\em Proceedings of the
  Twenty-Third Annual {ACM} Symposium on Principles of Distributed Computing,
  {PODC} 2004, St. John's, Newfoundland, Canada, July 25-28, 2004}, pages
  275--282. {ACM}, 2004.

\bibitem{LamportSP82}
Leslie Lamport, Robert~E. Shostak, and Marshall~C. Pease.
\newblock The byzantine generals problem.
\newblock {\em {ACM} Trans. Program. Lang. Syst.}, 4(3):382--401, 1982.

\bibitem{leighton1994packet}
Frank~Thomson Leighton, Bruce~M Maggs, and Satish~B Rao.
\newblock Packet routing and job-shop scheduling ino (congestion+ dilation)
  steps.
\newblock {\em Combinatorica}, 14(2):167--186, 1994.

\bibitem{MaurerT12}
Alexandre Maurer and S{\'{e}}bastien Tixeuil.
\newblock On byzantine broadcast in loosely connected networks.
\newblock In {\em Distributed Computing - 26th International Symposium, {DISC}
  2012, Salvador, Brazil, October 16-18, 2012. Proceedings}, pages 253--266,
  2012.

\bibitem{Nayak0SVX20}
Kartik Nayak, Ling Ren, Elaine Shi, Nitin~H. Vaidya, and Zhuolun Xiang.
\newblock Improved extension protocols for byzantine broadcast and agreement.
\newblock In {\em 34th International Symposium on Distributed Computing, {DISC}
  2020, October 12-16, 2020, Virtual Conference}, pages 28:1--28:17, 2020.

\bibitem{parter2019small}
Merav Parter.
\newblock Small cuts and connectivity certificates: A fault tolerant approach.
\newblock In {\em 33rd International Symposium on Distributed Computing (DISC
  2019)}. Schloss Dagstuhl-Leibniz-Zentrum fuer Informatik, 2019.

\bibitem{ParterYPODC19}
Merav Parter and Eylon Yogev.
\newblock Secure distributed computing made (nearly) optimal.
\newblock In {\em Proceedings of the 2019 {ACM} Symposium on Principles of
  Distributed Computing, {PODC} 2019, Toronto, ON, Canada, July 29 - August 2,
  2019}, pages 107--116, 2019.

\bibitem{pease1980reaching}
Marshall Pease, Robert Shostak, and Leslie Lamport.
\newblock Reaching agreement in the presence of faults.
\newblock {\em Journal of the ACM (JACM)}, 27(2):228--234, 1980.

\bibitem{Pelc92}
Andrzej Pelc.
\newblock Reliable communication in networks with byzantine link failures.
\newblock {\em Networks}, 22(5):441--459, 1992.

\bibitem{pelc1996fault}
Andrzej Pelc.
\newblock Fault-tolerant broadcasting and gossiping in communication networks.
\newblock {\em Networks: An International Journal}, 28(3):143--156, 1996.

\bibitem{PelcP05}
Andrzej Pelc and David Peleg.
\newblock Broadcasting with locally bounded byzantine faults.
\newblock {\em Inf. Process. Lett.}, 93(3):109--115, 2005.

\bibitem{pelc2007feasibility}
Andrzej Pelc and David Peleg.
\newblock Feasibility and complexity of broadcasting with random transmission
  failures.
\newblock {\em Theoretical Computer Science}, 370(1-3):279--292, 2007.

\bibitem{Peleg:2000}
David Peleg.
\newblock {\em Distributed Computing: A Locality-sensitive Approach}.
\newblock SIAM, 2000.

\bibitem{santoro1989time}
Nicola Santoro and Peter Widmayer.
\newblock Time is not a healer.
\newblock In {\em Annual Symposium on Theoretical Aspects of Computer Science},
  pages 304--313. Springer, 1989.

\bibitem{SantoroW90}
Nicola Santoro and Peter Widmayer.
\newblock Distributed function evaluation in the presence of transmission
  faults.
\newblock In {\em Algorithms, International Symposium {SIGAL} '90, Tokyo,
  Japan, August 16-18, 1990, Proceedings}, pages 358--367, 1990.

\bibitem{spielman09}
Daniel~A. Spielman.
\newblock Lecture notes in spectral graph theory (lecture 19).
\newblock \url{http://www.cs.yale.edu/homes/spielman/561/2009/lect19-09.pdf},
  Fall 2009, Yale University.

\bibitem{toueg1987fast}
Sam Toueg, Kenneth~J Perry, and TK~Srikanth.
\newblock Fast distributed agreement.
\newblock {\em SIAM Journal on Computing}, 16(3):445--457, 1987.

\bibitem{Upfal94}
Eli Upfal.
\newblock Tolerating a linear number of faults in networks of bounded degree.
\newblock {\em Inf. Comput.}, 115(2):312--320, 1994.

\bibitem{weimann2013replacement}
Oren Weimann and Raphael Yuster.
\newblock Replacement paths and distance sensitivity oracles via fast matrix
  multiplication.
\newblock {\em ACM Transactions on Algorithms (TALG)}, 9(2):1--13, 2013.

\bibitem{wulff2017fully}
Christian Wulff-Nilsen.
\newblock Fully-dynamic minimum spanning forest with improved worst-case update
  time.
\newblock In {\em Proceedings of the 49th Annual ACM SIGACT Symposium on Theory
  of Computing}, pages 1130--1143, 2017.

\end{thebibliography}

\newpage
\appendix

\section{Missing Proofs of Section~\ref{sec:broadcast-algo}}\label{sec:miss} \label{app:known}
\APPENDPRFACT

\end{document}